\documentclass[11pt]{article}
\usepackage{amsmath}
\usepackage{amssymb}
\usepackage{amsthm}
\usepackage{graphicx}
\usepackage[arrow,curve,matrix,arc,2cell]{xy}
\UseAllTwocells
\usepackage[utf8]{inputenc}
\usepackage[unicode]{hyperref}
\usepackage[numbers,sort]{natbib}	
\hypersetup{
    pdftoolbar=true,        
    pdfmenubar=true,        
    pdffitwindow=false,     
    pdfstartview={FitH},    
    pdftitle={Internal multi-portfolio rebalancing processes: Linking resource allocation models and biproportional matrix techniques to portfolio management},    
    pdfauthor={Kelli Francis-Staite},     
    pdfsubject={Internal multi-portfolio rebalancing processes},   
    pdfcreator={Kelli Francis-Staite},   
    pdfproducer={Kelli Francis-Staite}, 
    pdfkeywords={Portfolio Rebalancing} {Superannuation} {Asset Allocation}{RAS algorithm}{Portfolio Management},  
    pdfnewwindow=true,      
    colorlinks=true,       
    linkcolor=black,          
    citecolor=black,        
    filecolor=black,      
    urlcolor=black 
    linkcolor=red,          
    citecolor=black,        
    filecolor=magenta,      
    urlcolor=cyan           
}

\usepackage{geometry}
\usepackage{color}
\usepackage{amsfonts}
\usepackage{stackrel}
\usepackage{verbatim}

\newcommand{\R}{\mathbb{R}}

\newcommand{\bcb}{\begin{color}{blue}}
\newcommand{\bcr}{\begin{color}{red}}
\newcommand{\bcg}{\begin{color}{cyan}}
\newcommand{\ec}{\end{color}}
\title{Multi-portfolio internal rebalancing processes\\\large Linking resource allocation models and biproportional matrix techniques to portfolio management}
\author{Kelli Francis-Staite\footnote{The University of Adelaide and Statewide Super}}
\date{\today}

\usepackage{cleveref}
\theoremstyle{plain}
\newtheorem{theorem}{Theorem}[section]
\newtheorem{lemma}[theorem]{Lemma}
\newtheorem{proposition}[theorem]{Proposition}
\newtheorem{corollary}[theorem]{Corollary}
\theoremstyle{definition}
\newtheorem{definition}[theorem]{Definition}
\newtheorem{example}[theorem]{Example}
\newtheorem{remark}[theorem]{Remark}
\numberwithin{figure}{section}

\RequirePackage{natbib}

\usepackage{algpseudocode}

\usepackage{algorithm}
\begin{document}

\maketitle
\begin{abstract}
    This paper describes multi-portfolio \emph{internal} rebalancing processes used in the finance industry. Instead of trading with the market to \emph{externally} rebalance, these internal processes detail how portfolio managers buy and sell between their portfolios to rebalance. We give an overview of currently used internal rebalancing processes, including one known as the \emph{banker} process and another known as the \emph{linear} process. We prove the banker process disadvantages the nominated banker portfolio in volatile markets, while the linear process may advantage or disadvantage portfolios.
    
    We describe an alternative process that uses the concept of \emph{market-invariance}. We give analytic solutions for small cases, while in general show that the $n$-portfolio solution and its corresponding `market-invariant' algorithm solve a system of nonlinear polynomial equations. It turns out this algorithm is a rediscovery of the RAS algorithm (also called the \emph{iterative proportional fitting procedure}) for biproportional matrices. We show that this process is more equitable than the banker and linear processes, and demonstrate this with empirical results. 
    
    The market-invariant process has already been implemented by industry due to the significance of these results. 
\end{abstract}
\tableofcontents
\section{Introduction} \label{sec:introduction}

This paper describes and extends multi-portfolio rebalancing processes currently used in the financial industry. Rebalancing processes are the strategy or rules that a portfolio manager undertakes to maintain their portfolios of assets around certain target allocations, as in \citet[pg.~167]{Bernstein2010}. In this paper, we will classify rebalancing processes into two types: the \emph{external} ones, where rebalancing of the portfolios occurs by trading with the market, and the \emph{internal} ones, where rebalancing of the portfolios occurs through reallocation of current assets between the different portfolios. The focus of this paper will be on internal processes, which are only applicable for multi-portfolio managers, although we survey both in \S\ref{sec: rebalintro}.

Internal rebalancing processes are a type of resource allocation problem. In general, resource allocation problems have been well studied in the literature, and are considered a core problem of economics as in \citet{Conrad1999}. Some questions that economists study include how resources are currently allocated, forecasting how they may be allocated in the future, and determining how best to allocate. Internal rebalancing processes consider the third point of how a portfolio manager should best allocate assets, although much of the connected literature considers the first two points, focusing on estimation and forecasting. 

Optimisation and game theory techniques may be used to consider how best to allocate resources as in \citet[\S5]{Nisan2007}. In the game theory setting, utility functions determine payoffs for individuals for allocating resources and solutions are usually in the form of equilibria. Such equilibra may not guarantee the optimal solution for each individual as in the Prisoner's Dilemma, see for example \citet[pg.~4]{LeytonBrown2008}. In contrast, the optimisation perspective defines a collective objective function that is required to be optimised subject to a set of constraints, as in \cite[\S II]{Gardner1990}. 

While many problems can be considered from both perspectives, our analysis of internal rebalancing processes in \S\ref{sec:rebalmulti} will start by examining the existing processes and their properties, such as the \emph{market-invariant} property. We will show that each rebalancing process can be formulated in terms of objective and utility functions in \S\ref{sec: MIoptimisation}, although this is not required to formulate these processes.

More specifically, internal rebalancing processes seek an allocation of assets to all portfolios such that all assets are allocated, all portfolios have the required total assets, and this allocation is as `close to' the target asset allocation (usually written as a matrix $M$) as possible. The asset allocation is usually written as a matrix $A$ and the target allocation as a matrix $M$. It is called a \emph{process} as at each point in time there may be a different amount of assets, a different total required for each portfolio and potentially a different target $M$, and one requires a systematic way of allocating the assets to the portfolios. This process then determines the definition of `close to'. 

It turns out there are many solutions to this problem, infinitely many as we discuss in \S\ref{sec:rebalmulti} depending on the notion of `close to'. We will show that the property of market-invariance will determine such a definition of `close to', giving us a well-defined process and an algorithm to calculate the allocation matrix $A$. This market-invariant property is particularly important for our analysis, as the internal rebalancing process with this property does not advantage or disadvantage portfolios due to market movements. These notions of `close to' usually determine an objective function and allow us to write these processes from an optimisation perspective, which we do in \S\ref{sec: MIoptimisation}.

While internal rebalancing problems are well known in practical settings of portfolio management, the theory behind them is quite general. For instance, finding such a matrix $A$ is a requirement of estimating input-output matrices as in \citet{Bacharach1971}, and a similar matrix must be found to describe powerflows along electricity lines as we discuss in \Cref{subsec:supplydemandgeneral} and \Cref{app: electricitynetworks}. These different applications usually determine what `close to' should mean.

After undertaking this research independently, we found that the market-invariance property is related to the problem in the economics literature of finding a \emph{biproportional} matrix as in \citet{Bacharach1965}, and the market-invariant algorithm is a rediscovery of the \emph{RAS algorithm} from \citet{Stone1942} (see also \citet{lahr2004} for a survey of this work). This algorithm is also known by various other names including the \emph{biproportional fitting algorithm}  and the \emph{iterative proportional fitting procedure} (see \citet{Lomax2015} or \citet{Lovelace2015} for example). We are not unique in this rediscovery, with early independent works by \citet{Kruithof1937} for telephone traffic, \citet{Deming1940} for census data, and Sheliekhovskii and Bregman as in \cite{Bregman1967} for convex optimisation. 

In particular, in \citet{Bacharach1971} the target allocation matrix $M$ is known as an \emph{input-output} matrix for a closed Leontief System (see \citet{Leontief1986}, \citet{Ryan1953}, and \citet{Berman1979}), and it was used to describe the structure of economic systems. In this case, instead of portfolios composed of assets, the system describes production of goods from different proportions of commodities. 

Taking this goods and commodities example from \cite{Bacharach1971}, the idea of finding a biproportional matrix is as follows. Over a given period of time there should be a matrix $A^*$ that describes how the commodities are being used by each industry to produce each good. However, there may be several unknowns --- the biproportional case considers that the total number of commodities being used and the total output of the commodities is known, along with the matrix $M$ which is some known measure of how much of each commodity is required for each good (which may be estimated for example from a previous period of time). The matrix $A^*$ is the true proportion used by each industry to produce each good, and this is unknown, and the problem becomes that of how to estimate $A^*$. The solution in \cite{Bacharach1971} proposes that the matrix $A^*$ can be estimated by a matrix of the form of a product $B_1MB_2$ where $B_1$ and $B_2$ are positive diagonal matrices. This matrix is then called \emph{biproportional} to $M$. 

We note that the literature seems to use these biproportional approaches for estimating or forecasting purposes. There is subsequent discussion how well these estimates fit, including how to generate confidence intervals, as in \citet[\S~1]{Bacharach1971}. Our problem is not to estimate or forecast a matrix, instead we are using these techniques to decide how best to allocate assets, so the discussion of confidence intervals and other measures of estimation do not apply here. However, these biproportional techniques and their extensions can be used to determine internal rebalancing processes. 

More recently, the literature has broadened these biproportional studies to consider when different information is known or unknown. For instance if one or more entries of the matrix $A$ are known or include negative elements, as in the GRAS algorithm from \citet{Junius2003} and its extensions detailed in \citet{Huang2008}, or one or more entries of the column or row totals $p_j$, $a_i$ are unknown, as in \citet{Temursho2021}. See also \citet[\S~7.4.1]{Miller2009} or \cite[Ch.~18]{Handbook2018} for a detailed summary. 

There are also multi-dimensional examples (that is, $m_1\times m_2\times \cdots \times m_k$ dimensional matrices and tensors) as studied in \citet{Sugiyama2017} and \citet{VJ2021}. We do not study these here, although these extensions can be meaningful for internal rebalancing processes. For example, the multi-dimensional case could be used to consider other types of asset exposures, including exposures across different currencies, countries or industries. We also discuss interpreting negative entries for internal rebalancing processes in \S\ref{subsubsec:negativity}.

In general, the literature on biproportional problems is extensive. There are theoretical studies as in \citet{Bregman1967}, \citet{Mesnard1994} and \citet{Bacharach1971}; there are motivations from statistics as in \citet{Kullback1968} for contingency tables and from economics as in \citet{Bacharach1965} for input-output models; and there are a wide range of applications. The research in this paper was undertaken without knowledge of this literature nor the techniques within, and has now been updated to include references to this research. 

This paper adds much to the literature. From the perspective of an interested portfolio manager, the work is entirely self-contained, surveying the existing portfolio management rebalancing practices, and discussing alternative techniques and results. 

In \S\ref{sec: rebalintro} we survey the existing literature on rebalancing process, particularly external rebalancing process, and discuss their importance for portfolio managers. While we would have preferred to focus on internal rebalancing process in this review, the author has found no existing literature that covers this other than the theoretical connection to the biproportional literature described above. In \S\ref{sec:rebalmulti} we seek to rectify this, where we define internal rebalancing problems in general and describe the processes used in practice. These include the banker process in \S\ref{subsec:banker}, the linear process in \S\ref{subsec:linear}, and hybrid processes in \S\ref{subsec:hybrid}. We also discuss the relationships to optimisation problems in \S\ref{subsec:optimisation} and how supply and demand problems, including electricity power flow constraints, can be considered as rebalancing processes in \Cref{subsec:supplydemandgeneral} and \Cref{app: electricitynetworks}. 

In \S\ref{sec: MarketinvariantRebal} we discuss the market-invariant rebalancing process. We initially give the definition in terms of the market-invariant property, and then we show how this results in a well defined process. We then determine analytic solutions in small cases including where the dimension of the matrix $M$ is $(m,n)=(2,2),(2,3),(3,2),(4,2),(2,4)$, then discuss what happens in the case $(m,n)=(3,3)$ and higher order cases in \S\ref{subsec:mn2} and \S\ref{subsec:mn234}. While the market-invariant process is related to the RAS algorithm, the presentation of defining it in terms of a property is novel and the analytic solutions have not appeared in the literature so far.

In general, we show that the market-invariant algorithm in \S\ref{subsec:generalcase} gives the process for $M$ of arbitrary dimension. This algorithm is a rediscovery of the RAS algorithm. We present our own proof of convergence. We then compare the market-invariant process to the banker and linear processes in \S\ref{subsec:comparisons}. We show that under certain market conditions the banker process disadvantages the banker portfolio and the linear process may advantage or disadvantage each portfolio, while the market-invariant process does not disadvantage any portfolio over another. This is in contrast to the general understanding in practice that these different rebalancing strategies have little effect on portfolio returns over time. This motivates using the market-invariant process for internal portfolio rebalancing to minimise any inequities from market movements. 

This research is a significant contribution to the financial industry, and has already been implemented in the Superannuation industry in Australia. Funds have responded to this research by discontinuing the use of banker and linear rebalancing process or hybrids, and instead using the market-invariant process to minimise inequity between portfolios due to market movements.

\section{Rebalancing processes in the finance industry}\label{sec: rebalintro}

A core problem for portfolio managers is to maintain exposures as markets fluctuate and as they contribute to or redeem from their investments. Some portfolio managers employ a `buy and hold' strategy, where after the initial purchase of assets they continue to hold the assets and never correct for market movements or cash flows (also called a \emph{drift-weight} portfolio as in \Citet{Granger2014}). However, portfolio managers usually have target allocations for their assets (or collections of assets called \emph{asset classes}) that contribute to maximising returns while achieving a certain level of risk, as first motivated by \Citet{Markowitz1952}.  Portfolio managers then undertake a process of \emph{rebalancing}, where they buy or sell assets to maintain their desired market exposures. 

An example to keep in mind would be a simple portfolio consisting of two different asset classes, say shares and bonds. There may be a target asset allocation of, say, $60\%$ shares (which we call a $60\%$ \emph{weight} to shares) while the rest is bonds. We would say that this portfolio aims to be  $60\%$ (directly) exposed to the share market and $40\%$ (directly) exposed to the bond market. We call the collection of shares the shares asset class, and similarly for bonds. If the shares' market value increases, the weight of the share asset class increases and we say that the portfolio is \emph{overweight} shares and \emph{underweight} bonds. At this point the portfolio manager holding this portfolio must make a decision whether or not to rebalance their portfolio back to the target asset allocations. To do so, they would need to sell shares and buy bonds by trading with the market.

A portfolio manager with a single portfolio may manage these rebalancing decisions manually or may use pre-determined rules to trade with the market to maintain their exposures close to their target allocations. This may improve risk characteristics as in \Citet{Ilmanen2015}. We call such rules \emph{(external) rebalancing processes}. 

Such rebalancing processes have been well studied throughout the literature. They are often extensions of `buy low, sell high' rules that give conditions on the extent of the buy and sell, and what is considered low or high. They also suggest how frequently to apply such rules particularly given the costs involved in their implementation. For example, \Citet{Zilbering2015} suggest rebalancing back to the target allocations whenever there are $5\%$ deviations or larger on an annual or semi-annual basis, and \citet[pg.~167]{Bernstein2010} suggests rebalancing even less frequently for tax reasons. While rebalancing is considered an important part of a portfolio manager's strategy as in \citet{Weinstein2003}, there does not appear to be any consensus on which processes to apply, with different strategies perform optimally in different conditions as in \citet{Tsai2001}.

More advanced rebalancing processes tend to use optimisation techniques (such as minimising cost and/or minimising (conditional) value-at-risk as in \citet{Maclean2006} and \citet{Meghwani2018}); consider derivatives as in \citet{Israelov2018}; use stochastic probability theory, dynamic programming and heuristic approximations such as machine learning algorithms as in \citet{Kritzman2008}, \citet{Perrin2020} and \citet{Sun2006}; or use fuzzy logic models as in \citet{Fang2006}. In general, these processes rely on the ability for a portfolio manager to be able to buy and sell with the market --- in particular, they rely on a reasonable level of \emph{liquidity} of their assets. Much of the statistical data required to implement these processes also relies on the ability of such data to be collected (or predicted) and may not have the same meaning for illiquid assets --- for example, how does one understand daily or weekly volatility of an asset class whose assets only price monthly or quarterly? See \citet[Ch.~13.4]{Ang2014} for details on issues with illiquid asset data. 

When a portfolio manager has more than one portfolio under management with potentially different target allocations, the literature has concentrated on managers pooling together the assets and applying these processes to the collective totals. More recently, the literature has begun to consider the interrelationships between the portfolios and how to distribute the assets and costs between the rebalanced portfolios `fairly', as in \Citet{Stubbs2007}, \Citet{Iancu2014} and \Citet{Zhang2019}.

The focus in this paper is not on external rebalancing processes, but instead on what we call \emph{internal} rebalancing processes, which exist for multi-portfolio managers. These rebalancing processes do not consider buying and selling with the market and are purely concerned with the distribution of assets (and costs) fairly to all portfolios. They do not currently rely upon historical or predicted statistics, such as expected returns or variances, and they resemble the resource allocation problems detailed in \S\ref{sec:introduction}. Such processes are used for funds with multiple portfolios that use portions of one or more of the same pooled assets in each of the portfolios' asset allocations. 

Extending our simple example above to consider this, say an investment manager is managing two portfolios, both worth $\$100$, one with a target exposure of $55\%$ shares and $40\%$ bonds and $5\%$ cash, and the other with a target exposure of $35\%$ shares and $55\%$ bonds and $10\%$ cash. To manage this, the portfolio manager thinks in aggregate: they need a total of $\$55+\$35=\$90$ in a shares asset class and $\$40+\$55=\$95$ in a bonds asset class. The investment manager has selected a collection of shares and bonds that they think will do well, so they buy a total of $\$90$ worth in shares and $\$95$ in bonds to form the asset classes. They then distribute $\$55$ worth of the shares asset class to the first portfolio and the rest to the second portfolio, and distribute $\$40$ of the bonds asset class to the first portfolio and the rest to the second portfolio, leaving the remaining $\$15$ as cash. 

Now say the second portfolio has a large withdrawal so that it is left with only $2\%$ cash, however at the same time the shares portfolio has a large distribution of cash from its investments. This would result in the first portfolio becoming overweight cash (and underweight shares and bonds) while the second portfolio is underweight cash (and over weight shares and bonds). Without trading with the market, the portfolio manager notices that they can redistribute the asset classes between the portfolios to correct for some of these overweights and underweights. How the manager decides to do this would require an internal rebalancing process.

While the portfolios in this example involved \emph{liquid}\footnote{By liquidity we mean \emph{market liquidity} as opposed to \emph{funding liquidity}. See \citet[pg.~5]{Mehrling2021} and \citet{Brunnermeier2008} for further details.} asset classes, where the assets can be bought or sold in the time required at the volume required without moving the price, internal rebalancing process are particularly important for portfolios with asset classes that are illiquid. These asset classes cannot be bought or sold in the short time frames desired by the asset managers. They are also important as buying or selling with the market incurs fees or buy/sell spreads that managers are often incentivised to reduce. 

The motivating example for this process comes from Superannuation funds in Australia, which invest on behalf of working Australians and provide a pension income for these Australian workers in retirement. They hold various portfolios with different mixes of asset classes to provide for different investment desires of their members. These funds are incentivised to hold illiquid assets to meet return and volatility targets, however must also minimise fees such as transaction costs.  
These cost pressures incentivise funds to minimise their external rebalances and their illiquid holdings prevent frequent rebalancing of these assets classes. However they can and do internally rebalance as frequently as needed, as this is comparatively cost-less and there are no issues for illiquid assets.  Superannuation funds usually undertake these internal rebalancing processes daily, weekly or monthly in a systematised way. 

In the next section we discuss rebalancing processes common in Superannuation funds, for which there are no published references. The author knows of no published research on these internal processes nor their properties, and this paper will begin to fill this gap in the literature.

\section{Internal rebalancing processes}\label{sec:rebalmulti}
We start this section with necessary definitions before describing internal rebalancing processes known in industry. We will generally drop the word `internal' as all rebalancing processes described in the following will be internal unless labeled by `external'.

\subsection{Definitions and feasibility}
Here we define rebalancing problems and processes, as well as discuss the feasibility of a rebalancing problem and the degrees of freedom that constrain the possible processes. 

\begin{definition}\label{defn: rebalancingproblem}
A \emph{rebalancing problem} $(M,a,p)$ consists of the following: 
\begin{itemize}
    \item  A non-negative column vector $a=(a_1,\ldots,a_m)^T\in \R^m$ (asset class totals) and a non-negative column vector $p=(p_1,\ldots,p_n)^T\in \R^n$ (portfolio totals) where $\sum_i a_i=\sum_j p_j $.
    \item A real non-negative $m\times n$ matrix $M$ (target asset allocation) with entries $M_{ij}$ and with columns sums $\sum_{i}M_{ij}=1$. Here, $M_{ij}$ is the target proportion of portfolio $j$ required from asset class $i$.
    \item The problem of finding a real non-negative $m\times n$ matrix $A$ (actual asset allocation as a proportion) with entries $A_{ij}$ and with column sums $\sum_{i}A_{ij}=1$ for all $j$ such that $Ap=a$ and $A$ is `close to' $M$.
\end{itemize}
For a given problem $(M,a,p)$ there might not exist an $A$ without violating one of the constraints. We call such a problem \emph{infeasible}, and otherwise it is called \emph{feasible}. We may apply additional constraints such as requiring $A_{ij}=0$ whenever $M_{ij}=0$, and we discuss when this is feasible in \Cref{lem:feasibilityRebal}. A \emph{rebalancing process} is a process used to find $A$ for any given feasible rebalancing problem $(M,a,p)$ that defines a notion of `close to'. That is, a rebalancing process is a function $(M,a,p)\mapsto A$ for feasible $(M,a,p)$.

While $A$ is the asset allocation as percentages, we will also consider the matrix $A^{\$}$ with entries \[A^{\$}_{ij}=A_{ij}p_j.\] This is the actual asset allocation as values (often monetary values) and we have that the conditions \[\sum_{i}A^{\$}_{ij}=p_j \quad \text{ and } \quad \sum_{j}A^{\$}_{ij}=a_i\] are equivalent to requiring $Ap=a$ and $\sum_{i}A_{ij}=1$ for all $j$. 

Note that given $A^{\$}$ we can scale each column $j$ to sum to 1 by dividing by $p_j$ to construct $A$ provided $p_j$ is positive. If $p_j$ is zero, we can set $A_{ij}=M_{ij}$ for all $i$.
\end{definition}

An immediate question that comes to mind is how many different rebalancing process are there? Or how many different ways to choose $A$ are available? We discuss this in the following remark. Note also that both $m>n$ and $n>m$ are possible and appear in practice.

\begin{remark} In general, a matrix $A$ with $m\times n$ elements has $mn$ degrees of freedom in the choice of the elements. In our case, the requirement that $Ap=a$ gives $m$ constraints and that $\sum_{i}A_{ij}=1$ for all $j=1,2,\ldots,n$ gives $n$ constraints to the matrix $A$. However, knowing that $\sum_ia_i=\sum_jp_j$ means that one of these constraints is redundant. This means we expect \[mn-n-m+1=(m-1)(n-1)\] degrees of freedom for the possible values of $A$. That is, in general we can determine $(m-1)(n-1)$ of the elements of $A$ before the constraints imply what the rest of the elements must be. 

The requirement that $A$ is non-negative means that we consider a subset (with non-empty interior) of these possible values, so we still have $(m-1)(n-1)$ degrees of freedom but on a smaller space. In general, other than the trivial case $m=1$ or $n=1$, this means we have an infinite number of solutions to the rebalancing problem. The requirement of finding an $A$ `close to' $M$ will mean that we want a process that can determine $A$ uniquely for each given $(M,a,p)$ within this infinite solution space. 

Say we have such a process, so that given any $(M,a,p)$ we can determine a unique $A$. Then consider that we can construct infinitely many more rebalancing processes by the following: for this given process and a given $(M,a,p)$ determine $A$ and $A^{\$}$ by the original process. Then take a matrix $E$ consisting of the elements $1, -1, 0$ such that all rows and columns sum to $0$ (if $A$ has any zero elements, require that $E$ has a zero in the same positions). Then take the number $b_A = \min_{i,j}\{A^{\$}_{ij}|A^{\$}_{ij}>0\}$ and a scaling element $\alpha \in [0,1]$. Then define the new process by constructing the allocation $\hat{A}^{\$}$ to be 
\[ \hat{A}^{\$} = A^{\$} + \alpha b_A E.\] As the rows and columns of $E$ sum to zero, it is follows that $\hat{A}^{\$}$ satisfies the constraints and that this is a valid rebalancing process. That $b_A$ is a minimum and $\alpha \in [0,1]$ ensures all elements remain non-negative. 

The choice of $\alpha\in[0,1]$ shows that given any rebalancing process we can construct infinitely many more rebalancing process, with at least one degree of freedom. In fact, $E$ can be chosen so that it has zeros in the same position as $A$ but otherwise so that it contains values between $-1$ and $1$ so that its maximum element is $1$ or its minimum element is $-1$ and each column and row sum to $1$. This gives significantly more degrees of freedom in the cases where $m,n\ge 2$ and either $m$ or $n$ is greater than $2$. 

For a given process, we usually require that $A=M$ whenever $Mp=a$ to be part of the definition of `close to'. Given such a process exists, we can still construct infinitely many other processes by defining $\alpha$ in the above to be a function of $A$, where say \[\alpha = \beta\max_{i,j}\{|A_{ij}-M_{ij}|\}, \quad \beta \in [0,1].\]

We may wish to add additional requirements such as imposing that $A_{ij}=0$ whenever $M_{ij}=0$. This will reduce the degrees of freedom depending on the number of zeros in $M_{ij}$ and knowing if this property holds beforehand can reduce the complexity of finding the solutions.
\end{remark}

We now discuss conditions on $M,a,$ and $p$ that ensure there are always feasible solutions to a rebalancing problem so that it is possible to construct a rebalancing process. We then give examples of rebalancing processes, showing that they do exist.

\begin{lemma} \label{lem:feasibilityRebal}
Let $(M,a,p)$ be a rebalancing problem. If there are no additional constraints on $A$ then all problems are feasible.

If we impose that $A$ must have zeros where $M$ has zeros then there are feasible solutions to this problem provided the following conditions hold:
\begin{itemize}

\item[(1)] For any given distinct list $J\subseteq \{1,\ldots,n\}$ then for each $j\in J$ we define the collection 
\[I_j=\{i \in \{1,\ldots,m\} : M_{ij}>0\}\] and we require that \[\sum_{i\in\cup_{j\in J} I_j}a_{i}\ge \sum_{j\in J} p_{j}.\] This ensures there are enough assets to allocate to the portfolio requirements. 
\item[(2)] Similarly for a given distinct list $I\subseteq \{1,\ldots,m\}$ then for each $i\in I$ we define the collection 
\[J_i=\{j \in \{1,\ldots,n\} : M_{ij}>0\}\] and we require that \[\sum_{j\in \cup_{i\in I}J_i}p_{i}\ge \sum_{i\in I} a_{i}.\] This ensures there is enough portfolio capacity to allocate all assets. 
\end{itemize}Note that these conditions hold trivially in the case that $M$ has no zero elements.
\end{lemma}

\begin{proof}
In the first case, we can define a rebalancing process that does not depend on $M$ at all and instead behaves in a `greedy' way. That is, we can define 
\[ A^{\$}_{ij} = \min\left\{a_i - \sum_{k=1}^{j-1}A^{\$}_{ik}, p_j - \sum_{k=1}^{i-1}A^{\$}_{kj}\right\}\ge 0,\] by starting at $i=j=1$, then iterating over $i$ and then incrementing $j$ and iterating again over $i$ etc. Here we define the sums to be zero if $j=1$ or $i=1$ respectively. By induction one can show that 
\[ \sum_{i=1}^k A^{\$}_{ij} = \min\left\{p_j, \sum_{i=1}^k (a_i - \sum_{s=1}^{j-1}A^{\$}_{is})\right\} \quad \forall k,\] and 
\[ \sum_{j=1}^k A^{\$}_{ij} = \min\left\{a_i, \sum_{j=1}^k (p_j - \sum_{s=1}^{i-1}A^{\$}_{sj})\right\} \quad \forall k.\]
Induction of the first equation on $j$ and the second one on $i$ gives that 
\[ \sum_{i=1}^m A^{\$}_{ij} = p_j\] for all $j$ and 
\[ \sum_{j=1}^n A^{\$}_{ij} =a_i\] for all $i$, so that this is a valid rebalancing process. We call it greedy as it allocates the largest amount possible for the first element and the subsequent largest amount possible in the next element, completely ignoring the information in $M$. 

In the case where $M$ has zeros, we require that $A^{\$}_{ij}=0$ whenever $M_{ij}=0$. Say an $(M,a,p)$ is feasible and such an $A^{\$}$ exists, however that the conditions in the lemma are not satisfied. Then without loss of generality there is a collection $J$ with corresponding $I_j$ such that 
\[\sum_{i\in\cup_{j\in J} I_j}a_{i}< \sum_{j\in J} p_{j}.\] Given such an $A^{\$}$ exists, we must have $\sum_{j=1}^nA^{\$}_{ij}=a_i$  and we also know that
\[\sum_{i\in J_i}A^{\$}_{ij}=\sum_{i=1}^mA^{\$}_{ij} = p_j.\] This gives 
\[\sum_{j\in J}\sum_{i\in I_j}A^{\$}_{ij}=\sum_{j\in J} p_j.\]
Putting these together we see that 
\begin{align*}
    \sum_{j\in J} p_j & = \sum_{j\in J}\sum_{i\in I_j}A^{\$}_{ij}\\
    & \le \sum_{j\in J}\sum_{i\in \cup_{j}I_j}A^{\$}_{ij}\\
    &= \sum_{i\in \cup_{j}I_j} \sum_{j\in J}A^{\$}_{ij}\\
    &= \sum_{i\in \cup_{j}I_j} a_i
\end{align*}
which is a contradiction. This means that for $A$ to exist in this case the conditions in the lemma must be satisfied. Conversely, if the conditions in the lemma are satisfied, we do have a feasible rebalancing problem and the greedy process can be modified to be as before however forcing the $i,j$-th element to be zero whenever $M_{ij}=0$. The greedy process then constructs a feasible solution $A$ to the rebalancing problem. 

Note that in the case that $M$ has all elements positive, then regardless of $J$ we have $I_j=\{1,\ldots,m\}$ for all $j\in J$ and regardless of $I$ we have $J_i=\{1,\ldots,n\}$ for all $i\in I$. This means the requirements in \Cref{lem:feasibilityRebal} reduce to simply checking that any subset of elements of $a$ must sum to something less than or equal to the sum of all elements of $a$, and any subset of elements of $p$ must sum to something less than or equal to the sum of all elements of $p$. Both of these conditions are trivially true due to the non-negativity of $a$ and $p$ which means any $(M,a,p)$ with $M$ having all entries positive is feasible. 
\end{proof}

In general, these conditions are not restrictive and matrices $M$ with zeros are usually feasible, although more zeros makes the problem more restrictive. In the sequel, we usually require that if $M$ has zeros then $A$ must have zeros in the same locations, and then feasible rebalancing problems $(M,a,p)$ satisfy conditions (1) and (2) from \Cref{lem:feasibilityRebal}.

\subsubsection{Non-negativity conditions} \label{subsubsec:negativity}
Before we proceed to give examples of rebalancing process that are used in practice, we should discuss the non-negativity conditions and, if we relaxed these conditions, what the interpretation of a negative solution would be in our context. We note that \citet[pg.~15]{lahr2004} suggest negative solutions in economics contexts are not easy to interpret in economic terms.

In the portfolio management context we can however interpret negative elements as follows. If one of the elements of $a$ contained a negative value, one could interpret this as a class of debt or liabilities to the portfolio holder rather than assets, for example a loan taken to fund the portfolio. If one of the elements of $p$ contained a negative value, one could interpret this to mean that this portfolio has been used to fund a loan (i.e. is a liability of someone, a form of \emph{leverage}) where they owe the contents of the portfolio plus the interest it would have earned on the underlying investments, which is currently attributed to the other portfolios. 

If one of the elements in $M$ (or equivalently in $A$) is negative, this would mean that the corresponding portfolio is \emph{leveraging} this asset class to increase its exposure to the other asset classes. For example, if the target allocation to portfolio $j$ was $110\%$ shares, and $-10\%$ bonds, then any positive movement for bonds would mean this portfolio paid the other portfolios the corresponding movements, and similarly positive movements for the shares would result in the other portfolios paying this portfolio the corresponding movements. 

While all of these negativity conditions are possible, in practice the last one would likely be the most useful in this setting, and can also arise from some of the rebalancing processes we detail below. Also relevant here is the Superannuation regulations which strictly prevent leveraging in almost all situations, as in the SIS Act \cite[\S67
\&\S97]{SISact}. However, leveraging is a common financial concept and therefore all of these interpretations have merit in the general portfolio management context.

We now detail specific rebalancing processes, beginning with the banker process.
\subsection{Banker rebalancing process} \label{subsec:banker}

This common rebalancing process selects a specific portfolio to be called the \emph{banker}. In this process, all portfolios except the selected `banker' portfolio are given their target asset allocations, and the banker portfolio is given the remaining funds. There are limitations on this, as the banker must be able to contain all the asset classes, and there must be enough of each asset class so that there are no negative exposures for the banker. This usually means the banker portfolio is chosen to be the largest portfolio in the fund.

In the case of Superannuation funds, they have default portfolios that often contain the largest proportion of their assets and this is usually the portfolio chosen to be the banker.

Specifically, this process works as defined in \Cref{alg:BankerRebal}.

\begin{algorithm}
\caption{Banker Rebalancing process} \label{alg:BankerRebal}
\begin{algorithmic} 
\Require Non-negative $n\times m$ matrix $M$ with columns summing to 1, non-negative vectors $a\in \R^n$ and $p\in \R^m$, with $\sum_ia_i=\sum_jp_j$. Banker identified as $j_1 \in \{ 1,2,\ldots, m\}$ with $p_{j_1}>0$.

\State \!\!\!\!\!\!{\bf Output:} Non-negative $n\times m$ matrix $A$ with columns summing to 1 such that $Ap=a$.

\Ensure $A$ is an $n\times m$ matrix of zeros.

\For{$i=1,2,\ldots,m$}
\For{$j=1,2,\ldots, j_1-1,j_1+1,\ldots,m$}

\State{$A_{ij} = M_{ij}$}

\EndFor

\State{$A_{ij_1} = \frac{a_i-\sum_{j=1, j\ne j_1}^nA_{ij}p_j}{p_{j_1}}$}

\If{$A_{ij_1}<0$}

\State{\bf Return} ``Problem is infeasible''

\EndIf

\EndFor

\State{\bf Return} $A$

\State{\bf Terminate algorithm}

\end{algorithmic}
\end{algorithm}

We can check that $Ap=a$ by the following calculation
\[ [Ap]_i = \sum_{j\ne j_1}M_{ij}p_j + \frac{(a_i-\sum_{j\ne j_1}M_{ij}p_j)}{p_{j_1}}p_{j_1} = a_i.\] We additionally require \[a_i\ge \sum_{j\ne j_1}M_{ij}p_j\] for each $i=1,\ldots, m$ for this process to be well defined and give a non-negative solution for $A$. This means there is enough of each asset class to satisfy the target allocation for all portfolios but the banker and for there to be a non-negative amount left over for the banker.

While there is usually a non-negative solution for $A$, extreme market events can cause this requirement to fail and the banker process would return an infeasible, negative solution. In practice funds usually manually adjust for this, either readjusting the asset allocations or changing the target allocations so that a feasible solution exists. However, if negative allocations are allowed, this process could be used without adjustment.

The amount $A_{ij}-M_{ij}$ is called the \emph{linear overweight} or \emph{underweight} for this asset allocation compared to the target for the asset class $i$ and portfolio $j$. The justification of this process is that the largest portfolio can absorb the overweights or underweights of an asset class more easily, and choosing the banker to be the largest portfolio means that the linear overweights for this asset class will be the smallest of all other possible choices of banker. However, this tends to ignore the target asset allocation of the banker altogether. 

It can also create inequity, as, for instance, if an asset class falls in value then applying this process will result in the banker effectively `selling' its assets in this asset class to the other portfolios, and vice versa when markets rise. This is contrary to the usual `buy low, sell high' external rebalancing process rules, which means that this process is disproportionately affecting the banker portfolio with respect to market movements. We discuss this further in \Cref{subsec:comparisons}. 

\subsection{Linear rebalancing process} \label{subsec:linear}
This process involves calculating the linear overweight of all of the portfolios compared to what they would be at their targets, and distributing this to each of the portfolios. No specified banker portfolio is required. The calculations are not complex, as in the following algorithm.

\begin{algorithm}
\caption{Linear Rebalancing process} \label{alg:LinearRebal}
\begin{algorithmic} 
\Require non-negative $n\times m$ matrix $M$ with columns summing to 1, non-negative vectors $a\in \R^n$ and $p\in \R^m$, with $\sum_ia_i=\sum_jp_j$.

\State \!\!\!\!\!\!{\bf Output:} Non-negative $n\times m$ matrix $A$ with columns summing to 1 such that $Ap=a$.

\Ensure $A$ is $n\times m$ matrix of zeros. Vector $y = Mp$, the vector of asset class totals if the fund was at the target asset allocation.

\State{Calculate $d=\frac{1}{\sum_{i=1}^n(a_i)}(a-y)$ the linear overweight/underweight}
\For{$i=1,2,\ldots,n$}
\For{$j=1,2,\ldots,m$}

\State{Set $A_{ij} = M_{ij}+d_i$}

\EndFor
\EndFor

\State{\bf Return} $A$

\State{\bf Terminate algorithm}

\end{algorithmic}
\end{algorithm}
This process distributes the linear overweights/underweights for each asset class compared to the target asset allocation $M$ to all portfolios. Here we have 
\[A_{ij} = M_{ij} + d_i, \quad \forall i,j,\] where $d_i$ is the linear overweight/underweight \[d_i = \frac{a_i-\sum_j M_{ij}p_j}{\sum_k a_k}.\]
We can check that $Ap=a$ by the following calculation
\[ [Ap]_i = \sum_{j}(M_{ij}+d_i)p_j =\sum_{j}M_{ij}p_j +\sum_j \frac{a_i - \sum_k M_{ik}p_k}{\sum_k a_k} p_j =  \sum_{j}M_{ij}p_j+a_i - \sum_{j}M_{ij}p_j = a_i,\] using that $\sum_ia_i=\sum_jp_j$. 

This process creates some sense of equity, as it means all portfolios receive these linear underweights and overweights evenly. However, it may be undesirable as it can mean some portfolios significantly change their allocations compared to others. 

For example, consider a portfolio that has a small allocation to an asset class, say $2\%$, compared to another portfolio that has a larger allocation at say $10\%$ to that asset class. If the asset class has a linear overweight of $2\%$, then the first portfolio will be rebalanced to an allocation of $4\%$ which is double its target asset allocation, while the second portfolio will be rebalanced to an allocation of $12\%$ which is only a $20\%$ increase. If the market then falls so that this asset class is at a neutral weight, the portfolios must be rebalanced to their target. 

In practice, this will have been achieved by the first portfolio `buying' from the the other portfolio at the top of the market for that asset class, and `selling' at the bottom of the market, contrary to usual `buy low, sell high' rules, and can disadvantage such a portfolio in the long run. The doubling of the exposure may also cause other asset allocation issues, such as significantly increasing the risk profile of some portfolios compared to others.

\subsection{Hybrid and proportional processes}\label{subsec:hybrid}
Many funds use hybrid banker and linear processes. Here, funds may group portfolios into categories, and, for example, apply the linear process to each category, and then select a banker portfolio within each category and apply the banker process within the categories. Another approach is to apply a process up to a certain limit to portfolios before resorting to a banker process, for example applying the linear process until a portfolio reaches a percentage linear overweight/underweight of a certain asset class (or group of asset classes, for example, a group consisting of all of the illiquid asset classes).

Funds have also considered trying to distribute the overweights/underweights proportionally instead of linearly. 
Here instead of the linear overweight/underweight $d_i$, the proportional overweight/underweight is calculated 
\[ q_i = \frac{a_i}{\sum_{j}M_{ij}p_j} ,\] and the aim is to give the portfolios the values $A_{ij} = M_{ij}q_i$. However, this does not ensure that $Ap=a$. Some funds have opted to avoid this issue by subsequently applying a banker process. In general however, the market-invariant rebalancing process we study in \S\ref{sec: MarketinvariantRebal} is a way to resolve this issue.

\subsection{Optimisation}\label{subsec:optimisation}
The rebalancing problem $(M,a,p)$ can be phrased as an optimisation problem. This could be the problem of minimising some objective function $D$ with output in $\R$ written in terms of $M$, $a$, $p$ and the output allocation $A$. This forms an optimisation program as follows
\begin{align*}
     &\min_{A} D(A,M,a,p)\\
     \text{such that}\quad &Ap = a,\\
     &\sum_iA_{ij}=1, \quad \forall j =1, \ldots, n,\\
     &A_{ij}\ge 0, \quad \forall i=1,\ldots, m,\;j =1, \ldots, n.
\end{align*} 
The choices of $D$ are numerous --- an easy example is to take $D$ to be the sum of the squared differences of the elements of $A$ and $M$. Solving by standard methods including Lagrange multipliers and convex optimisation as in \citet{Boyd}, or Newton-Raphson methods as in \cite[Chpt.~4]{Bonnans2003}, is usually feasible for different choices of $D$ and they can improve on the issues mentioned in the banker and linear rebalancing process. 

The previous two processes can be converted into optimisation problems. In the banker process, one choice of objective function which would give the same solution would be 
\[D(A,M,a,p) = \sum_{i,j\ne j_b} (A_{ij}-M_{ij})^2.\] Here, any metric to compute the distance between $A_{ij}$ and $M_{ij}$ would do, as given that the banker $j_b$ is not included in the sum, the optimisation process would set $A_{ij}=M_{ij}$ for all portfolios except the banker. For the linear case, the optimisation problem could have objective function 
\[D(A,M,a,p) = \sum_{i,j} \left(A_{ij}-M_{ij} - \frac{a_i-\sum_{k}M_{ik}p_k}{\sum_ka_k}\right)^2,\] and here the optimisation process would set $A_{ij}-M_{ij} - \frac{a_i-\sum_{k}M_{ik}p_k}{\sum_ka_k} = 0$ for each $i,j$. This means that $A_{ij}=M_{ij}+d_i$ for each $i,j$ for the linear overweight/underweight $d_i$.

In general, is there a natural choice of objective function and would we interpret it? The squared differences are a usual first guess in optimisation as it can be interpreted as a least-squares estimator. This and several extensions are discussed in \citet[\S4]{lahr2004}, noting that there are interpretations as the Person's $\chi^2$ statistic or as the Neyman's $\chi^2$ statistic, and \citet[\S2.6]{Senata2006} discusses the squared differences as a second order Taylor series approximation. These objective functions make sense when the focus is to estimate or forecast $A$, but have little intuitive meaning for our rebalancing processes.

As mentioned in the introduction, market data such as momentum overlays and risk measures including Variance at Risk (VaR) used in  \cite{Stubbs2007}, \cite{Iancu2014} and \cite{Zhang2019} could also be introduced to determine an optimisation problem, but we do not do this here. We will return to the question of what might be a natural choice of objective function in \S\ref{sec: MIoptimisation}. A summary of other optimisation techniques documented in the literature is available in \cite[Ch.~18]{Handbook2018}.

\begin{remark} 
\label{subsec:supplydemandgeneral}
The requirement that $Ap=a$ means this problem is a supply equal to demand problem. An interpretation of this is that the supply of each asset class is sufficient to meet the total demand from the portfolio, and that all assets are fully allocated. There are many such problems that require supply to equal to demand. 

One example is from electricity networks, where supply of electricity must equal demand of electricity and certain properties of powerlines determine how much electricity flows along each line. We detail this further in \Cref{app: electricitynetworks}.

We note that input-output models in economics also fall into this category as we mentioned in \S\ref{sec:introduction}. In particular, \citet[\S 3.8]{Bacharach1971} discusses such supply-demand constraints for network theory, and the corresponding optimisation approaches used to find solutions. 
\end{remark}

\section{The market-invariant rebalancing process} \label{sec: MarketinvariantRebal}

We now present the market-invariant rebalancing process. The aim of this process is to minimise or indeed eliminate any inequitable treatment that may occur between portfolios due to market fluctuations. For example, if a market shock event occurred where only one asset class declined in value, then each time the rebalancing process was undertaken no portfolio should need to buy or sell between other portfolios to balance this, which may advantage or disadvantage portfolios upon further market movements. Note that such a market event would be an opportunity for an external rebalance instead.

In the following sections, we will use the operation $\odot$ to mean multiplication element-wise, for example 
\begin{align*}
    \begin{bmatrix}a_1 & a_2\\ a_3 & a_4\end{bmatrix}\odot \begin{bmatrix}b_1 & b_2\\b_3 & b_4\end{bmatrix} = \begin{bmatrix} b_1a_1  & b_2a_2\\b_3a_3 & b_4a_4\end{bmatrix},
\end{align*}
while the operation $\oslash$ will mean to divide element-wise, for example
\begin{align*}
    \begin{bmatrix}a_1 & a_2 \\ a_3 & a_4\end{bmatrix}\oslash \begin{bmatrix}b_1 & b_2\\b_3 & b_4\end{bmatrix} = \begin{bmatrix}a_1 & a_2\\ a_3 & a_4\end{bmatrix}\odot \begin{bmatrix}\frac{1}{b_1} & \frac{1}{b_2}\\ \frac{1}{b_3} & \frac{1}{b_4}\end{bmatrix} = \begin{bmatrix} \frac{a_1}{b_1}  & \frac{a_2}{b_2}\\ \frac{a_3}{b_3} & \frac{a_4}{b_4}\end{bmatrix}.
\end{align*}
In practice we will want to be able to take a row vector $a$ and a matrix $B$ and multiply each column of $B$ by the corresponding element in $a$, then we will write this as follows
\begin{align*} (e_m^Ta) \odot B &= \begin{bmatrix} 1 &1 &\ldots& 1\end{bmatrix}^T\begin{bmatrix} a_1 &a_2 &\ldots& a_k\end{bmatrix} \odot \begin{bmatrix}b_{11} & b_{12} &\ldots & b_{1k} \\b_{21} & b_{22}& \ldots & b_{2k}\\ \vdots &&\ddots&\\b_{m1} &b_{m2}& \ldots & b_{mk}  \end{bmatrix} \\
&=\begin{bmatrix}a_{1} & a_{2} &\ldots & a_{k} \\a_{1} & a_{2}& \ldots & a_{k}\\ \vdots &&\ddots&\\a_{1} &a_{2}& \ldots & a_{k}  \end{bmatrix}\odot \begin{bmatrix}b_{11} & b_{12} &\ldots & b_{1k}\\ b_{21} & b_{22}& \ldots & b_{2k}\\ \vdots &&\ddots&\\b_{m1} &b_{m2}& \ldots & b_{mk}  \end{bmatrix}\\
&=\begin{bmatrix}b_{11}a_1 & b_{12}a_2 &\ldots & b_{1k}a_k \\b_{21}a_1 & b_{22}a_2& \ldots & b_{2k}a_k\\ \vdots &&\ddots&\\b_{m1}a_1 &b_{m2}a_2& \ldots & b_{mk}a_k  \end{bmatrix}.
\end{align*}
Here $e_m$ is a row vector of length $m$ with entries all ones. Note that $ae_k^T = \sum_i{a_i}$ and $a^Te_k\odot B$ would result in each row of $B$ being multiplied by the corresponding element of $a$ provided the number of rows in $B$ is the same as the number of element of $a$ (so $m=k$ in the example above). 

We now define the property of \emph{market-invariance}, which precisely captures the idea that no portfolio should need to buy or sell from other portfolios due to market movements.
\begin{definition} \label{def:marketinvar}
Let $(M,a,p)$ be a rebalancing problem. We say that a rebalancing process has the property of \emph{market-invariance} if it produces the asset allocation $A$ such that
\begin{itemize}
    \item[(1)] If $Mp = a$ then $A=M$; and
    \item[(2)] Say $A$ solves the rebalancing problem so that $Ap=a$. Let $(x_1,\ldots, x_m)^T\in (0,\infty)^m$ represent a proportional change in the value of assets $a=(a_1,\ldots, a_m)^T$ to 
    \[a_x = a\odot x =(x_1a_1,\ldots, x_ma_m)^T,\] 
    and corresponding changes to the portfolio values where $p$ becomes $p_x$ with
    \[p_{x,j}=p_j\sum_{i}A_{ij}x_i =[(A\odot (xe_n))p]_j.\] 
    Then the rebalancing process applied to the rebalancing problem $(M,a_x,p_x)$ must give the asset allocation
    \[A^{\$}_{ x} =A^{\$} \odot (xe_n) = (pe_m)^T\odot A \odot(xe_n) \quad \text{ and } A_x = (pe_m)^T\odot A \odot(xe_n) \oslash (p_xe_m)^T \]
    (where here $e_k=(1,\ldots, 1)$ is a row vector with $k$ elements). As indices we have 
    \[A_{x,ij} = \frac{x_iA_{ij}p_j}{p_{x,j}} \quad \text{ and } \quad A^{\$}_{x,ij} = x_jA_{ij}p_j .\]    
\end{itemize} 
\end{definition}

From the outset, this definition seems to give us neither existence nor uniqueness of the solution $A$, or even a way to find $A$. In \Cref{defn:market-invariantrebal} we will show this definition is equivalent to solving a system of polynomial equations, so that solving these polynomial equations for a non-negative solution will determine $A$. In \Cref{thm: market-invariantSequence} and \Cref{cor: market-invariantsequence} we will show that these equations always have a non-negative solution, so existence is guaranteed, and we see this explicitly for $m=n=2$ in \S\ref{subsec:mn2}. Finally, we see that \Cref{thm: bachuniqueness} determines uniqueness. This means this property uniquely determines a rebalancing process, which we call the \emph{market-invariant} rebalancing process. It turns out that this process also gives a way to distribute asset overweights/underweights proportionally, which we mentioned in \S\ref{subsec:hybrid}.

\begin{remark} \label{rem: equirelation}
The property of market-invariance can be used to define an equivalence relation $\sim$ as follows for $m\times n$ matrices. We say for $m\times n$ matrices $B_1$ and $B_2$ that $B_1\sim B_2$ if there are $m\times m$ and $n\times n$ diagonal matrices $R$ and $S$ respectively, with positive diagonal, such that $RB_1S = B_2$. We can see this relation is reflexive, that is $B_1\sim B_1$, by setting $R$ and $S$ to the identity matrices; it is reflexive  (if $B_1 \sim B_2$ then $B_2 \sim B_1$) by considering $R^{-1}$ and $S^{-1}$; and it is transitive by considering if $B_2 = R_1B_1S_1$, and $B_3 = R_2B_2 S_2$ then $B_3 = R_1R_2B_2S_2S_1$, so that $B_1\sim B_2$ and $B_2\sim B_3$ implies $B_1\sim B_3$. 

In our case, we have $A\sim A^{\$}_{ x}$, where we take $R$ to be the $m\times m$ diagonal matrix with $x$ along the diagonal, and $S$ to be the diagonal $n\times n$ matrix with $p$ along its diagonal. Similarly $A\sim A_x$, with the same $R$ however $S$ has elements $\frac{p_j}{p_{jx}}=[p\oslash p_x]_j$ on its diagonal. As the relation is an equivalent relation, there is a set of all $m\times n$ matrices that are related to $A$ (called an \emph{equivalence class}). We will use this equivalence relation to discuss how the property of market-invariance gives a unique set of polynomial equations as in \Cref{defn:market-invariantrebal}.
\end{remark}

\begin{proposition}\label{defn:market-invariantrebal}
For each rebalancing problem $(M,a,p)$, if a rebalancing process has the property of market-invariance from \Cref{def:marketinvar} then determining asset allocation $A$ is equivalent to finding non-negative column vectors $x'\in (0,\infty)^m$, $p'\in(0,\infty)^n$, $a'\in(0,\infty)^m$ to give  \[A_{ij}=\frac{x_i'M_{ij}p_j'}{p_j} \quad \text{ and } \quad A^{\$}_{ij} =x_i'M_{ij} p_j'.\] Here we require that $a'=Mp'$ and $p',x'$ solve the following system of polynomial equations, divided into the asset class equations
\begin{align}
f_1&=x_1'(M_{11}p_1'+\ldots+M_{1n}p_n')=a_1,\nonumber\\
&\vdots\nonumber\\
f_m&=x_m'(M_{m1}p_1'+\ldots+M_{mn}p_n')=a_m, \label{eqn:assetclassequations}
\end{align}
and the portfolio equations
\begin{align}
g_1&=p_1'(M_{11}x_1'+\ldots+M_{m1}x_m')=p_1,\nonumber\\
&\vdots\nonumber\\
g_n&=p_n'(M_{1n}x_1'+\ldots+M_{mn}x_m')=p_n. \label{eqn:portfolioequations}
\end{align}
As matrices, we can rewrite these equations as 
\begin{align*}
x'\odot Mp'&=a\\
p'\odot M^tx'&=p.
\end{align*} 
We can see that any non-negative solution $p',a'$ to \Cref{eqn:assetclassequations} and \Cref{eqn:portfolioequations} can be extended to a 1-dimensional space of solutions, $p'(t)$ and $x'(t)$, where $t\in (0,\infty)$ is a scaling factor. That is, $p'(t) = \frac{1}{t} p'$, and $x'(t) = tx$. These all have the same corresponding asset allocation $A$. 

Note that negative solutions can be created by taking $t\in (-\infty,0)$, although this still results in the same value of $A$. Also, these equations are not independent, as we have $\sum_ia_i=\sum_jp_j$, so one equation can be removed from the set \Cref{eqn:assetclassequations} and \Cref{eqn:portfolioequations} without changing the solutions. 

\end{proposition} 

In the literature, the matrix $A^{\$}$ is called \emph{biproportional} to $M$ (or to $M\odot (pe_m)^T$), as in \citet[pg.~296--297]{Bacharach1965}.
\begin{proof}
Say we have a rebalancing problem $(M,a,p)$ and we do not have that $Mp=a$. We want to be able to use both points of \Cref{def:marketinvar} to determine $A$.

Say there exists another set of portfolio values $p'$ and asset class values $a'$ such that $Mp'=a'$, and we also have asset class movements $(x_1',\ldots,x_m')=x'$ such that 
\[a'_{x'}=a'\odot x'= Mp'\odot x'= a \quad \text{ and }\quad  p=p'\odot M^Tx'=p'_{x'}.\] 
Then we can solve $(M,a',p')$ using the solution $A'=M$ by \Cref{def:marketinvar}(1) and then we can solve $(M,a,p)$ by setting 
\[A^{\$}=A'\odot x'e=(p'e_m)^T\odot M\odot x'e_n, \quad \text{ so } \quad A^{\$}_{ij}=x'_iM_{ij}p'_j\]
and 
\[A=A'\odot x'e=(p'e_m)^T\odot M\odot x'e_n\oslash (pe_m)^T, \quad \text{ so } \quad A_{ij}=\frac{x'_iM_{ij}p'_j}{p_j}\]
using \Cref{def:marketinvar}(2). 
This means we need to find $x',p',a'$ such that $Mp'=a'$ and 
\[ \sum_ix_i'M_{ij}p_j' = p_j \quad \text{ and }  \sum_jx_i'M_{ij}p_j'=a_i,\] which are the asset class and portfolio equations. This means that if \Cref{def:marketinvar} gives a well defined rebalancing process implies there exists $a',x',p'$ that solve the asset class and portfolio equations.

For the reverse argument, it is clear that if $Mp=a$ then $x'=(1,\ldots, 1)$, $p'=p$, $a'=a$ solve the asset class equations, and if this is unique then the asset class and portfolio equations imply \Cref{def:marketinvar}(1). Using the notation from \Cref{def:marketinvar}(2), if $A$ solves $(M,a,p)$ using the asset class and portfolio equations, then $A^{\$} =RMS$ where $R$ has $x'$ along the diagonal and $S$ has $p'$ along the diagonal as in \Cref{rem: equirelation}, so $M\sim A^{\$}$. Then if $A_x$ solves $(M,a_x,p_x)$ using the asset class and portfolio equations using $x'',p'',a''$, then $M\sim A^{\$}_{x}$ where $R$ has $x''$ along the diagonal and $S$ has $p''$. As we have shown that $\sim$ is an equivalence relation, we must have that $A^\$\sim A^{\$}_{x}$. This means that being able to solve the asset class and portfolio equations for a unique $A$ implies that the property of market-invariance from \Cref{def:marketinvar} gives a well defined rebalancing process.

This means that \Cref{def:marketinvar} is well defined if and only if for each rebalancing problem $(M,a,p)$ there exists such $a',p',x'$ that give the solution $A$ uniquely. This is equivalent to saying that for each feasible $(M,a,p)$ there exists a unique element in the $\sim$-equivalence class of $M$ that has row sums equal to $a$ and column sums equal to $p$. 
\end{proof}

This result implies that our market-invariant rebalancing process is well defined provided there exists a (unique up to scaling) solution $x',p',a'$ that solves the system of equations in \Cref{defn:market-invariantrebal}. However, systems of polynomial equations are notoriously hard to solve analytically. They often reduce to solving a single high order polynomial equation, and the Abel-Ruffini theorem (see \cite{Ayoub1982,Cox2012}) tells us that there is no general analytic solution for polynomial equations above order 4. Instead, numerical methods such as Newton-Raphson methods (see \cite[Chpt.~4]{Bonnans2003}) can be used to solve for solutions where they exist.

In the cases $(m,n)=(2,2),(2,3),(3,2),(4,2)$ we can write down an analytic solution, which we do next, as well as consider the case $(m,n)=(3,3)$ and higher order cases. We then proceed to prove such solutions exist in general and give an algorithm to find these solutions. We also remark on uniqueness of solutions.

\subsection{\texorpdfstring{Solution for $(m,n)=(2,2)$}{Solution for (m,n)=(2,2)}} \label{subsec:mn2}

Here we determine the analytic formula for the market-invariant process in the case where the dimension of $M$ is $(m,n)=(2,2)$.

\begin{proposition}
Let $(M,a,p)$ be a feasible rebalancing problem. When $(m,n)=(2,2)$ the asset class and portfolio equations from \Cref{defn:market-invariantrebal} result in a unique solution $A$ to this rebalancing problem, and a one-dimensional solution set $x',p'$. 
\end{proposition}
\begin{proof}
If an element of $M$ is zero, the solution is trivial. For example, if $M_{11}=0$, then $A^{\$}_{11}=0$, $A^{\$}_{21}=p_1$, $A^{\$}_{12}=a_1$, $A^{\$}_{22}=a_2-a_1$, and $A_{ij}=A^{\$}_{ij}/(\sum_{j}A^{\$}_{ij})$. Then we can define $x'$ such that $x'_2 =\frac{A^{\$}_{21}}{M_{21}}$, and $p'$ such that $p'_j= \frac{A^{\$}_{2j}}{M_{2j}x_2}$ for $j=1,2$, which implies $p'_1=1$, and $x_1 =\frac{A^{\$}_{12}}{M_{12}p'_2}$. Multiplying $p'$ by $t$ and $x'$ by $1/t$ for $t\in \R\setminus{0}$ gives the full solution set for values of $p'$ and $x'$, while the unique positive solution occurs for all values of $t\ge 0$. 

If we assume that all elements of $M$ are positive then we can consider the asset class equations
\begin{align*}
f_1&=x_1'(M_{11}p_1'+M_{12}p_2')=a_1,\\
f_2&=x_2'(M_{21}p_1'+M_{22}p_2')=a_2,
\end{align*}
and the portfolio equations
\begin{align*}
g_1&=p_1'(M_{11}x_1'+M_{21}x_2')=p_1,\\
g_2&=p_2'(M_{12}x_1'+M_{22}x_2')=p_2.
\end{align*}

We can write
\begin{align} \label{eqn: xs}
x_1'=\frac{a_1}{M_{11}p_1'+M_{12}p_2'},\\\nonumber
x_2'=\frac{a_2}{M_{21}p_1'+M_{22}p_2'},
\end{align}

then we find that 
\begin{align*}
g_1&=p_1'\left(M_{11}\frac{a_1}{M_{11}p_1'+M_{12}p_2'}+M_{21}\frac{a_2}{M_{21}p_1'+M_{22}p_2'}\right)=p_1,\\
g_2&=p_2'\left(M_{12}\frac{a_1}{M_{11}p_1'+M_{12}p_2'}+M_{22}\frac{a_2}{M_{21}p_1'+M_{22}p_2'}\right)=p_2.
\end{align*}
 We can rewrite as 
\begin{align*}
\frac{p_1'}{p_1}\left(M_{11}a_1(M_{21}p_1'+M_{22}p_2')+M_{21}a_2(M_{11}p_1'+M_{12}p_2')\right)=(M_{11}p_1'+M_{12}p_2')(M_{21}p_1'+M_{22}p_2'),\\
\frac{p_2'}{p_2}\left(M_{12}a_1(M_{21}p_1'+M_{22}p_2')+M_{22}a_2(M_{11}p_1'+M_{12}p_2')\right)=(M_{11}p_1'+M_{12}p_2')(M_{21}p_1'+M_{22}p_2')
\end{align*}

Both equations have the same right-hand sides, so we can equate their left hand sides as follows
\begin{align*}
\frac{p_1'}{p_1}\left(M_{11}a_1(M_{21}p_1'+M_{22}p_2')+M_{21}a_2(M_{11}p_1'+M_{12}p_2')\right)=\\
\frac{p_2'}{p_2}\left(M_{12}a_1(M_{21}p_1'+M_{22}p_2')+M_{22}a_2(M_{11}p_1'+M_{12}p_2')\right)
\end{align*}

So this is now just a polynomial of degree two in our two unknowns $p_1'$ and $p_2'$. Expanding and collecting like terms we have 

\begin{align*}
(p_1')^2M_{11}M_{21}\frac{a_1+a_2}{p_1}- (p_2')^2M_{12}M_{22}\frac{a_1+a_2}{p_2}+&\\p_1'p_2'\left(M_{11}M_{22}(\frac{a_1}{p_1}-\frac{a_2}{p_2})+M_{21}M_{12}(\frac{a_2}{p_1}-\frac{a_1}{p_2})\right) &=0.
\end{align*}

Applying the quadratic formula we have 

\[p_1'=p_2'\left(\frac{-b\pm \sqrt{b^2-4dc}}{2d}\right)\]
where 
\begin{align*}
    b &= M_{11}M_{22}(\frac{a_1}{p_1}-\frac{a_2}{p_2})+M_{21}M_{12}(\frac{a_2}{p_1}-\frac{a_1}{p_2})\\
    c &=M_{12}M_{22}\frac{a_1+a_2}{p_2}\\
    d &= M_{11}M_{21}\frac{a_1+a_2}{p_1}
\end{align*}

So we have that one of $p_1'$ and $p_2'$ is free, and the other is a multiple of the first. 

We now show that we always get a unique positive solution for $A$ and $A^{\$}$, despite having a free variable. 
We have that 
\[x_i' = \frac{a_i}{M_{i1}p'_1+M_{i2}p_2'} = \frac{a_i}{p_2'(M_{i1}K_{\pm}+M_{i2})}  \quad \text{where} \quad K_{\pm} = \frac{-b\pm \sqrt{b^2-4dc}}{2d},\] so that 
\[A^{\$}_{ij}=x_i'M_{ij}p_j' =\frac{a_iM_{ij}p_j'}{p_2'(M_{i1}K_{\pm}+M_{i2})} =\begin{cases} \frac{a_iM_{ij}K_{\pm}}{M_{i1}K_{\pm}+M_{i2}} & j =1,\\
\frac{a_iM_{ij}}{M_{i1}K_{\pm}+M_{i2}} & j=2. \end{cases} \]
This means that $A^{\$}$ is independent of the choice of $p_1'$ or $p_2'$, and only depends on the choice of sign for $K_+$ or $K_{-}$

We will show only $K_{+}$ results in a non-negative solution for $A^\$$. Note that the discriminant $b^2-4dc$ is non-zero unless one or more elements of $M$ are zero, which reduces back to the first case. We have that $-dc>0$ regardless of $M,p,a$, and then we have $b\le |b| \le \sqrt{b^2} \le \sqrt{b^2-4dc}$, so that $-b- \sqrt{b^2-4dc}<0$ and $-b+ \sqrt{b^2-4dc}>0$, regardless of the values of $M,p,a$. This means we always have a unique positive solution, corresponding to $-b+ \sqrt{b^2-4dc}$ and $K_{+}$, giving the unique solution to the rebalancing problem. We note that \[A_{ij}=x_i'M_{ij}=\frac{A^{\$}_{ij}}{\sum_kA^{\$}_{kj}}.\]

\end{proof}

\begin{example}
Say we have the rebalancing problem of $(M,a,p)$ where
\[ M= \begin{bmatrix} 0.3 & 0.5\\ 0.7 & 0.5\end{bmatrix}, \quad p=\begin{bmatrix} 120\\180 \end{bmatrix}, a=\begin{bmatrix} 100\\200 \end{bmatrix}.\]  Using the market-invariant rebalancing process for $n=m=2$ we have solution 
\begin{align*} A^{\$} = \begin{bmatrix} 27.1003& 72.8997\\92.8997 & 107.1003\end{bmatrix},\quad &  A = \begin{bmatrix} 0.2258 & 0.7742\\0.4050 & 0.5950\end{bmatrix}\\ p'(t) = \frac{1}{t}\begin{bmatrix} 0.6196\\1 \end{bmatrix},\quad &x'(t) = t\begin{bmatrix} 145.7995\\214.2005\end{bmatrix}.\end{align*}

If we include the negative solutions as well, we see that there is another solution for $A$ corresponding to the negative sign in front of the square root in the scaling factor between $p_1'$ and $p_2'$. This solution is 
\begin{align*} A^{\$} = \begin{bmatrix} -332.1003& 432.1003\\452.1003 & -252.1003\end{bmatrix},\quad & A = \begin{bmatrix}-2.7675 & 2.4006\\3.7675 & -1.4006\end{bmatrix}, \\
p'(t) = \frac{1}{t}\begin{bmatrix} -1.2810\\1 \end{bmatrix}, \quad&x'(t) = t\begin{bmatrix} 864.2005\\-504.2005\end{bmatrix}.\end{align*}
For both $A$ solutions, we can allow $t$ to vary in $\R\setminus\{0\}$ and this means the solutions $p'(t)$ give two lines through the origin with the origin removed. This is a real smooth algebraic variety with four disconnected components each isomorphic to $(0,\infty)$ and these correspond to each quadrant in $\R^2$.

We attach the Matlab code for these calculations in \Cref{app:code}.
\end{example}

\subsection{More small cases} \label{subsec:mn234}
We can extend the results of the previous section to more small cases. 

\begin{proposition}
Let $(M,a,b)$ be a feasible rebalancing problem. When $(m,n) = (2,3), (3,2), (4,2),$ or  $(2,4)$ solutions to the asset class and portfolio equations in \Cref{defn:market-invariantrebal} can be written analytically. 
\end{proposition}
\begin{proof}
Starting with the case $(m,n)=(2,3)$, we will assume all $M_{ij}$ are non-zero else this reduces to the previous case. Note that the case $(m,n)=(3,2)$ is analogous.

We can consider the asset class equations
\begin{align*}
f_1&=x_1'(M_{11}p_1'+M_{12}p_2' + M_{13}p_3')=a_1,\\
f_2&=x_2'(M_{21}p_1'+M_{22}p_2' + M_{23}p_3')=a_2,
\end{align*}
and the portfolio equations
\begin{align*}
g_1&=p_1'(M_{11}x_1'+M_{21}x_2')=p_1,\\
g_2&=p_2'(M_{12}x_1'+M_{22}x_2')=p_2,\\
g_2&=p_3'(M_{13}x_1'+M_{23}x_2')=p_2.
\end{align*}
Given our previous discussions, to solve for positive solutions we can assume the scale parameter $x_2$ is equal to 1. Also, given the linear dependence of one of the equations, we will exclude $f_2$ from the calculations. 

Now, rewrite the portfolio equations as follows
\begin{align*}
    p_1' &= \frac{p_1}{M_{11}x_1'+M_{21}},\\
    p_2' &= \frac{p_2}{M_{12}x_1'+M_{22}},\\
    p_3' &= \frac{p_3}{M_{13}x_1'+M_{23}},
\end{align*}
and substitute these into $f_1$ to give 
\[ \frac{x_1'M_{11}p_1}{M_{11}x_1'+M_{21}} + \frac{x_1'M_{12}p_2}{M_{12}x_1'+M_{22}} + \frac{x_1'M_{13}p_3}{M_{13}x_1'+M_{23}} = a_1.\]
Rearranging and factorising in terms of $x_1'$ we have the following cubic equation
\begin{align*}
    B_1{x'}_1^3 + B_2{x'}_1^2 + B_3 x'_1 + B_4 =0
\end{align*}
with coefficients
\begin{align*}
    B_1 & = M_{11}M_{12}M_{13}( p_1+p_2+p_3 - a_1) = M_{11}M_{12}M_{13}a_2 > 0\\
    B_2 & = M_{11}M_{12}M_{23}(p_1+p_2-a_1) + M_{11}M_{22}M_{13}(p_1+p_3-a_1) + M_{21}M_{12}M_{13}(p_2+p_3-a_1)\\
    B_3 & = M_{11}M_{22}M_{23}(p_1-a_1) + M_{21}M_{12}M_{23}(p_2-a_1) + M_{21}M_{22}M_{13}(p_3-a_1)\\ 
    B_4 &= - a_1M_{21}M_{22}M_{23} <0.
\end{align*}
We can now apply the cubic formula (see for example \citet{Neumark1965}) for solving cubic polynomials in terms of $x_1'$ as follows. Define 
\begin{align*}
    \Delta_0 &= B_2^2 - 3B_1B_3,\\
    \Delta_1 &= 2B_2^3 - 9B_1B_2B_3 + 27B_1^2B_4,\\
    Q &=  \sqrt[\leftroot{-2}\uproot{2}3]{\frac{\Delta_1\pm \sqrt{\Delta_1^2 - 4\Delta_0^3}}{2}},\\
    \zeta &= \frac{-1+\sqrt{-3}}{2},\\
    y_k &=\begin{cases} -\frac{B_2}{3B_1} &\text{if } \Delta_0=\Delta_1=0 \\
    -\frac{1}{3B_1}\left(B_2 + \zeta^k Q + \frac{\Delta_0}{\zeta^kQ}\right), & \text{otherwise} 
    \end{cases}\qquad \text{for } k\in \{0,1,2\}.
\end{align*}
Then the solutions to the cubic are $x'_1=y_0,y_1,y_2$, and we can use the previous equations to find $p_1',p_2',p_3'$. If we use $x'_2=-1$, then all of the above still works however the formulae for $B_2$ and $B_4$ must be multiplied by $-1$. 

As cubics always contain one real root, then there are at least two real solutions corresponding to $x_2'=1$ and $x_2'=-1$. There may be up to $6$ solutions, however there may be fewer real solutions. For example, if $M_{ij}=0.5$ for all $i,j$, $a_1=a_2$, and $p_1=p_2=p_3=\frac{2}{3}a_1$, then we find 
\[ B_1 = B_2 = -B_3 = -B_4 = \frac{a_1}{8}\] and the cubic for $x_2'=1$ becomes 
\[ \frac{a_1}{8} (x_1'^3+x_1'^2-x_1'-1) = \frac{a_1}{8}(x_1'-1)(x_1'+1)^2.\] This has a double root for $x_1'=-1$, and a single root $x_1'=1$, which gives a positive solution. In \S\ref{subsec:generalcase} we show that there always exists a positive real root for the general $(m,n)$ case, so this formula must give this solution, and we have observed this in practice with the positive real root appearing for $x_1'=1$, and solution $y_1$.

The case $(m,n)=(2,4)$ (and analogously $(4,2)$) is very similar, however uses the quartic formula. Applying the same method results in a quartic 
\begin{align*}
    B_1{x'_1}^4 + B_2{x'_1}^3 + B_3 {x'_1}^2 + B_4x'_1+B_5 =0
\end{align*}
with coefficients
\begin{align*}
    B_1 & = M_{11}M_{12}M_{13}M_{14}( p_1+p_2+p_3 +p_4 - a_1) = M_{11}M_{12}M_{13}a_2 > 0\\
    B_2 & = M_{11}M_{12}M_{13}M_{24}(p_1+p_2+p_3-a_1) + M_{11}M_{12}M_{23}M_{14}(p_1+p_2+p_4-a_1) \\
    &+ M_{11}M_{22}M_{13}M_{14}(p_1+p_3+p_4-a_1)+M_{21}M_{12}M_{13}M_{14}(p_2+p_3+p_4-a_1)\\
    B_3 & = M_{11}M_{12}M_{23}M_{24}(p_1+p_2-a_1) + M_{11}M_{22}M_{13}M_{24}(p_1+p_3-a_1) + M_{21}M_{12}M_{13}M_{24}(p_2+p_3-a_1)\\
    &+M_{11}M_{22}M_{23}M_{14}(p_1+p_4-a_1) + M_{21}M_{12}M_{23}M_{14}(p_2+p_4-a_1) + M_{21}M_{22}M_{13}M_{14}(p_3+p_4-a_1)\\
    B_4 & = M_{11}M_{22}M_{23}M_{24}(p_1-a_1) + M_{21}M_{12}M_{23}M_{24}(p_2-a_1) \\
    &+ M_{21}M_{22}M_{13}M_{24}(p_3-a_1)+M_{21}M_{22}M_{23}M_{14}(p_4-a_1)\\
    B_5 &= - a_1M_{21}M_{22}M_{23}M_{24} <0.
\end{align*}
We can now apply the quartic formula, where we define
\begin{align*}
    \Delta_0 &= B_3^2-3B_2B_4+12B_1B_5\\
    \Delta_1 &= 2B_3^3-9B_2B_3B_4+27B_2^2B_5+27B_1B_4^2-72B_1B_3B_5\\
    r &= \frac{8B_1B_3-3B_2^2}{8B_1^3}\\
    s& = \frac{B_2^3-4B_1B_2B_3 + 8B_1^2B_4}{8B_1^3}\\
    T &=\frac{1}{2}\sqrt{-\frac{2}{3}r + \frac{1}{3B_1}\left(Q+\frac{\Delta_0}{Q}\right)} \\
    Q &= \sqrt[\leftroot{-2}\uproot{2}3]{\frac{\Delta_1+ \sqrt{\Delta_1^2 - 4\Delta_0^3}}{2}}\\
    y_{\pm,\pm}& = -\frac{B_2}{4B_1}\pm T\pm \frac{1}{2}\sqrt{-4T^2-2r+\frac{s}{T}}
\end{align*}
The solutions are $x_1'=y_{++},y_{+-},y_{--},y_{-+}$. Similarly if we consider the $x_2'=-1$ case, the formulae for $B_1, B_3, B_5$ must be multiplied by $-1$.

Again, if $M_{ij}=0.5$ for all $i,j$, $a_1=a_2$, and $p_1=p_2=p_3$, then we find 
\[ 2B_1 = B_2 = -B_4 = -2B_5 = \frac{a_1}{16}, \quad B_3 =0\] and the cubic for $x_2'=1$ becomes 
\[ \frac{a_1}{16} (x_1'^4+2x_1'^3-2x_1'-1) = \frac{1}{8}(x_1'-1)(x_1'+1)^3.\] 
This has a triple root for $x_1'=-1$, and a single root $x_1'=1$, which gives the positive solution we are looking for. Again \S\ref{subsec:generalcase} we show that there always exists a positive real root for the general $(m,n)$ case, so this formula must give this solution. 
\end{proof}

The next smallest case is $(m,n)=(3,3)$. However, unless there are one or more zeros in $M$, this in general results in a quintic polynomial or higher, for which there is no analytic form of the solution.

\subsubsection{Further comments on hybrid solutions}
While in the next section we will describe an algorithm to find a non-negative solution to the equations in \Cref{defn:market-invariantrebal}, we note that these analytic solutions can be used to form approximate solutions to this. For example, if we want to use the $(m,n)=(2,2)$ formula only, we can aggregate the portfolios and asset classes into groups. This works as follows. 

Set $I_1, I_2$ to be a non-empty partition of the asset classes $1,2,\ldots, m$, and $J_1$ and $J_2$ a non-empty partition of the portfolios $1,2 \ldots, n$. 
Set 
\[\hat{M}_{I_k,J_{\ell}} = \frac{\sum_{i\in I_k}\sum_{j\in J_\ell} M_{ij}p_j}{\sum_{j\in J_\ell}p_j}, \quad i=1,2, \quad j=1,2,\] which is like an aggregated target allocation for that group of asset classes and portfolios. Similarly aggregate the portfolio totals to 
\[\hat{p}=\left(\sum_{j\in J_1}p_j,\sum_{j\in J_2}p_j\right) \quad \text{ and } \quad \hat{a} = \left(\sum_{i\in I_1}a_i,\sum_{i\in I_2}a_i\right).\] 
Then apply the analytic rebalancing process solution for $(m,n)=(2,2)$ to $(\hat{M},\hat{a},\hat{p})$. Call the resulting asset allocation $A^1$ and then define
\[p^{k\ell}_j=A^1_{k\ell}p_j, \text{ for } j\in J_{\ell}, \quad \text{ and } a^{k\ell}_i =A^1_{k\ell}a_i \text{ for }i\in I_k; \quad \text{ for } k=1,2, \ell=1,2.\]
This allocates the proportions of $p$ and $a$ to each group. We then apply this process again, but now we partition each $I_k$ and $J_\ell$ further into two groups. If any partition only contains one element, we do not partition this further. We repeat this process until all partitions contain only one element, resulting in a final proportion for each asset class to each portfolio. 

This gives an approximate solution to the general case, although the choices of partitions give different solutions. Care must be taken if there are zeros in $M$, and this must be adjusted for appropriately, by, for instance, picking partitions that isolate the zeros first. In general, this approach might be most useful if there are certain asset classes that have high correlation, for example two categories of shares. Then picking partitions that groups these asset classes together can increase the effectiveness of the approximation, although we do not prove this here. In general, any rebalancing solutions can be integrated in this way to create new rebalancing processes, and at each time step a different process can be used if required.

\subsection{General case} \label{subsec:generalcase}

In this section we prove that for any given feasible rebalancing problem $(M,a,p)$ there exists a natural solution $A$ and $x',p'$ that solve the asset class and portfolio equations in \Cref{defn:market-invariantrebal}. This means that the market-invariant rebalancing process is well defined. From this, we describe an algorithm that can be implemented efficiently to find a unique solution for any rebalancing problem. 

Here we will often replace the index $i=1,\ldots,m$ with $s$ and the index $j=1,\ldots, n$ with $t$ when there are multiple summations present.


An intuitive way to interpret this next theorem is as follows. Say we have a feasible rebalancing problem $(M,a,p)$ and we want to distribute the asset overweights/underweights proportionally. 

We suggested in \S\ref{subsec:hybrid} that the proportional overweight/underweight 
\[ q_i = \frac{a_i}{\sum_{j}M_{ij}p_j} ,\]  
should be used to distribute assets to each portfolio, where portfolio $i$ should receive 
\[A^{\$}_{ij} = M_{ij}p_jq_i = \frac{M_{ij}p_ja_i}{\sum_{t}M_{it}p_t}.\] 
However, this does not ensure that the row sums of $A^{\$}$ are equal to $a$, although the column sums are equal to $p$. The following theorem is a way to use this intuition but to ensure that $Ap=a$, and it turns out that this gives a non-negative solution to \Cref{defn:market-invariantrebal}.

\begin{theorem} \label{thm: market-invariantSequence}
Construct two sequences of matrices $M^{(k,p)}$, $M^{(k,a)}$, such that 
\[ M^{(1,p)}_{ij}= M_{ij}p_j, \quad M^{(k+1,p)}_{ij}= \frac{M_{ij}^{(k,p)} p_j a_i}{\left(\sum_t M^{(k,p)}_{it}\right) \left(\sum_s\frac{M_{sj}^{(k,p)} a_i}{\sum_tM_{st}^{(k,p)}}\right)}\] 
and 
\[ M^{(1,a)}_{ij}= \frac{M_{ij}a_ip_j}{\sum_t M_{it}p_t}, \quad M^{(k+1,a)}_{ij}= \frac{M_{ij}^{(k,a)} a_i p_j}{\left(\sum_s M^{(k,a)}_{sj}\right) \left(\sum_t\frac{M_{it}^{(k,a)} p_t}{\sum_sM_{st}^{(k,a)}}\right)}\] for $k=1,2,\ldots$. We have the relationship \[ M_{ij}^{(k+1,p)} = \frac{M_{ij}^{(k,a)}p_j}{\sum_s M_{sj}^{(k,a)}}, \quad M_{ij}^{(k,a)} = \frac{M_{ij}^{(k,p)}a_i}{\sum_t M_{it}^{(k,p)}}.\] At each step we have $\sum_iM_{ij}^{(k,p)}=p_j$ and $\sum_j M_{ij}^{(k,a)}=a_i$, and as $\sum_ia_i=\sum_jp_j$ we have that \[\sum_{i,j}M_{ij}^{(k,p)}=\sum_{i,j}M_{ij}^{(k,a)} = \sum_ia_i=\sum_jp_j.\]

The sequences $M^{(k,p)}$, $M^{(k,a)}$ converge to $M^p$ and $M^a$ respectively, and we have $A^{\$}=M^p=M^a$. We define $x'_i =\frac{A^{\$}_{i1}}{M_{i1}}$ for all $i=1,\ldots, m$, and $p'_j= \frac{A^{\$}_{ij}}{M_{ij}x_i}$ for $j=1,\ldots,n$ which implies $p'_1=1$, and $A_{ij}=\frac{A^{\$}_{ ij}}{\sum_sA^{\$}_{ sj}}$. Then $A$ is a solution to the rebalancing problem $(M,a,p)$ and $x',p'$ satisfy the equations in \Cref{defn:market-invariantrebal}. 
\end{theorem}

\begin{corollary} \label{cor: market-invariantsequence}
The market-invariant rebalancing process from \Cref{defn:market-invariantrebal} applied to a rebalancing problem $(M,a,p)$ has a natural non-negative solution $(A,x',p')$ with $p_1'=1$. 
\end{corollary}

\begin{proof}
Proof of convergence of $M^{(k,p)}$ is as follows. Convergence of $M^{(k,a)}$ is analogous. 

We show that $\sum_jM_{ij}^{(k,p)}\rightarrow a_i$ for each $i$. We see that 
\begin{align}\nonumber
    \frac{\sum_jM_{ij}^{(k+1,p)}}{a_i} &= \frac{1}{\sum_j M_{ij}^{(k,p)}} \sum_j\frac{M_{ij}^{(k,p)}p_j}{\sum_s\frac{M_{sj}^{(k,p)}a_s}{\sum_tM_{st}^{(k,p)}}}\\\nonumber
    & \ge \frac{1}{\sum_j M_{ij}^{(k,p)}} \left(\sum_{j\ne j_i'}\frac{M_{ij}^{(k,p)}p_j}{\sum_sM_{sj}^{(k,p)}\max_s\frac{a_s}{\sum_tM_{st}^{(k,p)}}} + \frac{M_{ij_i'}^{(k,p)}p_{j_i'}}{\sum_sM_{sj_i'}^{(k,p)}\frac{a_s}{\sum_tM_{st}^{(k,p)}}}\right)\\
    &=\frac{1}{\sum_j M_{ij}^{(k,p)}} \left(\sum_{j\ne j_i'}\frac{M_{ij}^{(k,p)}}{\max_s\frac{a_s}{\sum_tM_{st}^{(k,p)}}} + \frac{M_{ij_i'}^{(k,p)}p_{j_i'}}{\sum_sM_{sj_i'}^{(k,p)}\frac{a_s}{\sum_tM_{st}^{(k,p)}}}\right). \label{eqn:proofeqt1}
\end{align}
Here we define $j_i'$ to be any $j$ for this $i$ such that $M_{ij}^{(k,p)}$ is non-zero, and we use that $\sum_iM^{(k,p)}_{ij} = p_j$ to go from the second line to the third.

Now consider $i'$ such that $\frac{a_i}{\sum_j M_{ij}^{(k,p)}}$ attains its minimum. Then we have 
\begin{align*}
    \sum_iM_{ij}^{(k,p)}\frac{a_i}{\sum_tM_{it}^{(k,p)}} &\le \sum_{i \ne i'} M_{ij}^{(k,p)} \max_s\frac{a_s}{\sum_t M_{st}^{(k,p)}} + M_{i',j}^{(k,p)} \min_s\frac{a_{s}}{\sum_t M_{st}^{(k,p)}}\\
    &=\sum_{i} M_{ij}^{(k,p)} \max_s\frac{a_s}{\sum_t M_{st}^{(k,p)}} - M_{i',j}^{(k,p)} (\max_s- \min_s)\frac{a_{s}}{\sum_t M_{st}^{(k,p)}}\\
    &=p_j \max_s\frac{a_s}{\sum_t M_{st}^{(k,p)}} - M_{i',j}^{(k,p)} (\max_s- \min_s)\frac{a_{s}}{\sum_t M_{st}^{(k,p)}}
\end{align*}
Substituting this into \Cref{eqn:proofeqt1} we have 
\begin{align}  
&\frac{\sum_jM_{ij}^{(k+1,p)}}{a_i}  \\\nonumber
&\ge \frac{1}{\sum_j M_{ij}^{(k,p)}} \left(\sum_{j\ne j_i'}\frac{M_{ij}^{(k,p)}}{\max_s\frac{a_s}{\sum_tM_{st}^{(k,p)}}} + \frac{M_{ij_i'}^{(k,p)}p_{j_i'}}{p_{j_i'} \max_s\frac{a_s}{\sum_t M_{st}^{(k,p)}} - M_{i',j_i'}^{(k,p)} (\max_s- \min_s)\frac{a_{s}}{\sum_t M_{st}^{(k,p)}}}\right)\\\nonumber
&=\frac{1}{\sum_j M_{ij}^{(k,p)}} \left(\left(\sum_jM_{ij}^{(k,p)}-M_{ij_i'}^{(k,p)}\right) \min_s\frac{\sum_tM_{st}^{(k,p)}}{a_s} + \frac{M_{ij_i'}^{(k,p)}p_{j_i'}\min_s\frac{\sum_t M_{st}^{(k,p)}}{a_s}}{p_{j_i'}  - M_{i',j_i'}^{(k,p)} \left(1- \frac{\min_s}{\max_s}\right)\frac{a_{s}}{\sum_t M_{st}^{(k,p)}}}\right)\\\nonumber
&=\frac{\min_s\frac{\sum_tM_{st}^{(k,p)}}{a_s}}{\sum_j M_{ij}^{(k,p)}} \left(\sum_jM_{ij}^{(k,p)}+M_{ij_i'}^{(k,p)}\left( \frac{p_{j_i'}}{p_{j_i'}  - M_{i',j_i'}^{(k,p)} \left(1- \frac{\min_s}{\max_s}\right)\frac{a_{s}}{\sum_t M_{st}^{(k,p)}}} -1\right) \right) \\ \label{eq:proofeqpentult}
&=\min_s\frac{\sum_tM_{st}^{(k,p)}}{a_s} \left(1+\frac{M_{ij_i'}^{(k,p)}}{\sum_j M_{ij}^{(k,p)}}\left( \frac{1}{1  - \frac{M_{i',j_i'}^{(k,p)}}{p_{j_i'}} \left(1- \frac{\min_s}{\max_s}\right)\frac{a_{s}}{\sum_t M_{st}^{(k,p)}}} -1\right) \right)\\
& \ge \min_s\frac{\sum_tM_{st}^{(k,p)}}{a_s}. \label{eq:proofeqtend}
\end{align}
As this is true for all $i$ this means that 
\[ \min_i\frac{\sum_jM_{ij}^{(k+1,p)}}{a_i} \le \min_i\frac{\sum_jM_{ij}^{(k,p)}}{a_i}\] and so for each $i$ $ \min_i\frac{\sum_jM_{ij}^{(k,p)}}{a_i}$ is an increasing sequence in $k$.

Similarly, we can swap all of the maximums and minimums to get a bound 
\begin{align*} \frac{\sum_jM_{ij}^{(k+1,p)}}{a_i} 
& \le \max_i\frac{\sum_jM_{ij}^{(k,p)}}{a_i}. 
\end{align*}
Then for all $j$,  $\max_i\frac{\sum_jM_{ij}^{(k,p)}}{a_i}$ is a decreasing sequence in $k$.

Iterating over $k$ we see by induction that 
\begin{equation}\label{eq:proofeqtsqueeze} \min_i\frac{\sum_jM_{ij}^{(k,p)}}{a_i}\le \min_i\frac{\sum_jM_{ij}^{(k+r,p)}}{a_i}\le \frac{\sum_jM_{ij}^{(k+r,p)}}{a_i} \le \max_i\frac{\sum_jM_{ij}^{(k+r,p)}}{a_i}\le \max_i\frac{\sum_jM_{ij}^{(k,p)}}{a_i} \end{equation}
for all $r=0,1, \ldots$. This means that $\max_i\frac{\sum_jM_{ij}^{(k+r,p)}}{a_i}$ is a decreasing sequence in $r$ bounded below and $\min_i\frac{\sum_jM_{ij}^{(k+r,p)}}{a_i}$ is an increasing sequence in $r$ that is bounded above, so both converge by the Monotone Convergence Theorem (see \cite[Chpt.1~\S 1~Cor~1.6]{Knapp2016} or any standard real analysis text book). 

As sequences to converge, by the Squeeze theorem for sequences (see for example \cite[Thm.~3.19]{Rudin1976} or \cite[Chpt.~1~\S~1 Prop.~1.7]{Knapp2016}) we must have that \Cref{eq:proofeqpentult} converges to \Cref{eq:proofeqtend}. This implies that
\begin{align*}
    1+&\frac{M_{ij_i'}^{(k,p)}}{\sum_j M_{ij}^{(k,p)}}\left( \frac{1}{1  - \frac{M_{i',j_i'}^{(k,p)}}{p_{j_i'}} \left(1- \frac{\min_s}{\max_s}\right)\frac{a_{s}}{\sum_t M_{st}^{(k,p)}}} -1\right) \rightarrow 1\\
    & \Leftrightarrow\frac{M_{ij_i'}^{(k,p)}}{\sum_j M_{ij}^{(k,p)}}\left( \frac{1}{1  - \frac{M_{i',j_i'}^{(k,p)}}{p_{j_i'}} \left(1- \frac{\min_s}{\max_s}\right)\frac{a_{s}}{\sum_t M_{st}^{(k,p)}}} -1\right) \rightarrow 0\\ 
    & \Leftrightarrow \frac{1}{1  - \frac{M_{i',j_i'}^{(k,p)}}{p_{j_i'}} \left(1- \frac{\min_s}{\max_s}\right)\frac{a_{s}}{\sum_t M_{st}^{(k,p)}}} -1 \rightarrow 0\\ 
    & \Leftrightarrow \frac{1}{1  - \frac{M_{i',j_i'}^{(k,p)}}{p_{j_i'}} \left(1- \frac{\min_s}{\max_s}\right)\frac{a_{s}}{\sum_t M_{st}^{(k,p)}}} \rightarrow 1\\
    & \Leftrightarrow 1  - \frac{M_{i',j_i'}^{(k,p)}}{p_{j_i'}} \left(1- \frac{\min_s}{\max_s}\right)\frac{a_{s}}{\sum_t M_{st}^{(k,p)}} \rightarrow 1\\ 
    & \Leftrightarrow \frac{M_{i',j_i'}^{(k,p)}}{p_{j_i'}} \left(1- \frac{\min_s}{\max_s}\right)\frac{a_{s}}{\sum_t M_{st}^{(k,p)}} \rightarrow 0\\
     & \Leftrightarrow \left(1- \frac{\min_s}{\max_s}\right)\frac{a_{s}}{\sum_t M_{st}^{(k,p)}} \rightarrow 0\\
     & \Leftrightarrow \left(\frac{\min_s}{\max_s}\right)\frac{a_{s}}{\sum_t M_{st}^{(k,p)}} \rightarrow 1.
\end{align*}
From this we can conclude that $\max_i\frac{\sum_jM_{ij}^{(k,p)}}{a_i}$ and $\min_i\frac{\sum_jM_{ij}^{(k,p)}}{a_i}$ must both converge to the same value, and by \Cref{eq:proofeqtsqueeze} and the Squeeze theorem for sequences again, this means $\frac{\sum_jM_{ij}^{(k+t,p)}}{a_i}$ also converges and converges to the same limit. 

Due the symmetry between $M^{(k,a)}$ and $M^{(k,p)}$ we have the corresponding result that $\max_j\frac{\sum_iM_{ij}^{(k,a)}}{p_j}$, $\min_j\frac{\sum_iM_{ij}^{(k,a)}}{p_j}$ and $\frac{\sum_iM_{ij}^{(k+t,a)}}{p_j}$ all converge to the same limit. 

Using this, we have 
\[  T \leftarrow \min_i\frac{\sum_jM_{ij}^{(k,p)}}{a_i}\le \frac{\sum_jM_{ij}^{(k,p)}}{a_i} \le \max_i\frac{\sum_jM_{ij}^{(k,p)}}{a_i} \rightarrow T \] for $T$ some constant which must not depend on $i$. 
Then 
\[ \sum_jM_{ij}^{(k,p)} \rightarrow a_iT.\] Taking the sum over $i$ we have 
\[ \sum_ia_i= \sum_i\sum_jM_{ij}^{(k,p)} \rightarrow \sum_ia_iT,\] and as there is no dependence on $k$ in the right hand or left hand side, we have that $T=1$. So indeed 
$\max_i\frac{\sum_jM_{ij}^{(k,p)}}{a_i}$, $\min_i\frac{\sum_jM_{ij}^{(k,p)}}{a_i}$, and $\frac{\sum_jM_{ij}^{(k+t,p)}}{a_i} $ converge to 1, and similar for the corresponding $(k,a)$ terms. Then we must have 

\[\sum_jM_{ij}^{(k,p)} \rightarrow a_i.\] Similarly \[\sum_iM_{ij}^{(k,a)} \rightarrow p_j.\] 
This in turn means that $M_{ij}^{(k,p)}$ and $M_{ij}^{(k,a)}$  must also converge, as they are defined by multiplying by $\frac{\sum_jM_{ij}^{p}}{a_i}$ and $ \frac{\sum_iM_{ij}^{a}}{p_j}$. 

Note that $\sum_i M^{(k,a)}_{ij} = p_j$ for all $j$ and $ \sum_j M^{(k,a)}_{ij} = a_i$ for all $i$. Then we have the difference 
\begin{align*}
    M_{ij}^{(k,a)} - M_{ij}^{(k,p)} &= \frac{M_{ij}^{(k,p)}a_i}{\sum_tM_{it}^{(k,p)}} - M_{ij}^{(k,p)}\\
    & = M_{ij}^{(k,p)}\left (\frac{a_i}{\sum_tM_{it}^{(k,p)}}-1 \right)\\
    &\rightarrow 0 \text{ as } k\rightarrow \infty.
\end{align*}
This means the limits are $M^a=M^p=A^{\$}$, and we know that $\sum_j M^{p}_{ij} = \sum_j M^{a}_{ij}=a_i$, and  $\sum_i M^{p}_{ij} =\sum_i M^{a}_{ij}= p_j$.

The proof that the limit solves the equations in \Cref{defn:market-invariantrebal} is by induction to find $x$ and $y$ as follows.

We will show for each $k$ there exist non-negative $x^k\in \R^n$, $y^k\in \R^m$  such that $M^{(k,p)}_{ij} = x_i^k M_{ij}y_j^k$ for all $i,j$. The initial case is $M^{(1,p)}_{ij}= M_{ij}p_j$, so this is true with $x_i^1=1, y_j^1=p_j$. 

Assume it is true for $M^{(k,p)}$ for some $x^k,y^k$. Then 
\begin{align*}
M^{(k+1,p)}_{ij}= \frac{M_{ij}^{(k,p)} p_j a_i}{\left(\sum_t M^{(k,p)}_{it}\right) \left(\sum_s\frac{M_{sj}^{(k,p)} a_s}{\sum_tM_{st}^{(k,p)}}\right)}\\
= \frac{x^k_ia_i}{\left(\sum_t M^{(k,p)}_{it}\right)} M_{ij} \frac{ p^k_j}{ \left(\sum_s\frac{M_{sj}^{(k,p)} a_s}{\sum_tM_{st}^{(k,p)}}\right)}
\end{align*}
as required, with 
\[ x_i^{k+1} = \frac{x_ia_i}{\left(\sum_t M^{(k,p)}_{it}\right)}, \quad y_j^{k+1} = \frac{ p_j}{ \left(\sum_s\frac{M_{sj}^{(k,p)} a_s}{\sum_tM_{st}^{(k,p)}}\right)}.\] 

A similar proof shows that there are $\hat{x}^k$ and $\hat{y}^k$ for $M^{(k,a)}$, with $\hat{x}^k_i = \frac{x^k_ia_i}{\sum_j M_{ij}^{(k,p)}}$, $x^{k+1}_i = \hat{x}^k_i$ and $\hat{y}^k_j = y^k_j$, $y^{k+1}_j= \frac{\hat{y}^k_jp_j}{\sum_iM_{ij}^{(k,a)}}$. We see that at each iteration $k$ then $x^k, y^k$ solve the portfolio equations and $\hat{x}^k$ and $\hat{y}^k$ solve the asset class equations. As $k\rightarrow \infty$ then we have $x^k$ and $\hat{x}^k$ converge to the same value $x$, and $y^k$ and $\hat{y}^k$ converge to the same value $y$. This convergence implies we solve both equations simultaneously, giving the solution to the equations in \Cref{defn:market-invariantrebal} as required. 

Finally we define $x'$ such that $x'_i =\frac{A^{\$}_{i1}}{M_{i1}}$ for all $i=1,\ldots, m$, and $p'$ such that $p'_j= \frac{A^{\$}_{ij}}{M_{ij}x_i}$ for $j=1,\ldots,n$ which implies $p'_1=1$. Then we have allocation $A$ with $A_{ij}=\frac{A^{\$}_{ ij}}{\sum_{i=1}^mA^{\$}_{ ij}}$. These $A,p',x'$ solve the rebalancing problem $(M,a,p)$ and satisfy the equations in \Cref{defn:market-invariantrebal}, so they are a solution to the market-invariant rebalancing process.
\end{proof}

This theorem shows the existence of solutions to \Cref{defn:market-invariantrebal} by constructing such a non-negative solution. We can then describe a rebalancing process by using this.

We can use this theorem to construct \Cref{alg:EqtRebal} that finds this natural solution and we call this the \emph{market-invariant rebalancing algorithm}. This algorithm directly calculates a finite number of elements in the sequence $M^{(1,p)}$, $M^{(1,a)}$, $M^{(2,p)}$, $M^{(2,a)}$ ,$\ldots,$ as described in \Cref{thm: market-invariantSequence}. As it only constructs a finite number of elements in the sequence, then it is unlikely to have fully converged upon termination. At this point, it takes the last element calculated in the sequence $M^{(r,p)}$ which has columns summing to $p$ and determines any difference between its row sums and $a$ (note that the sum of these differences is $0$). It then adds these into $M^{(r,p)}$ in a way proportion to $p$. This then enforces that the columns still sum to $p$ and now the rows also sum to $a$. Note that the differences could be added back to $M^{(r,p)}$ in many different ways, for example they could be all added to the largest portfolio, just like the banker process.

\begin{algorithm}
\caption{Market-invariant rebalancing algorithm} \label{alg:EqtRebal}
\begin{algorithmic} 
\Require A rebalancing problem $(M,a,p)$ with non-negative $n\times m$ matrix $M$ with columns summing to 1, non-negative column vectors $a\in \R^n$ and $p\in \R^m$. Number of iterations $r$ or accuracy level $q$ for the monetary amount (e.g. accurate to $q=0.01$). We will use $e_k=(1,\ldots,1)$ as the row vector all ones of dimension $k$ to be used in this algorithm.

\State \!\!\!\!\!\!{\bf Output:} Non-negative $n\times m$ matrix $A$ with columns summing to 1 such that $Ap=a$ up to some accuracy. Here $A$ solves the rebalancing problem $(M,a,p)$ and is an approximate solution to the market-invariant rebalancing process.

\Ensure If any $p$ or $a$ are zero, set the allocation $A^{\$}$ to be zero for the corresponding rows or columns, and let $A=M$ for those rows and columns. Reduce the problem $(M,a,p)$ to no longer contain these rows or columns. Then set  \[M^{(1,p)}=M\odot e_m^Tp^T; \quad \text{ and } \quad  M^{(1,a)}=M^{(1,p)}\odot [a \oslash (M^{(1,p)}e_n^T)]e_n.\]

\State{\bf Step 1}

\For{$k=1,2,\ldots,$}

$M^{(k+1,p)}=M^{(k,a)}\odot [p \oslash ((M^{(k,a)})^Te_m^T)]e_m$

$M^{(k+1,a)}=M^{(k+1,p)}\odot [a \oslash (M^{(k+1,p)}e_n^T)]e_n$

\If{$k+1=r$ \text{ OR }  $\max_{ij}(M^{(k+1,a)} - M^{(k+1,p)})_{ij}\le q$ }

\State {\bf end for, go to Step 2}
\EndIf
\EndFor

\State{\bf Step 2}

Set $r$ to be the last value of $k$ used.

Calculate the vector $d = a-M^{(r+1),p}e_n^T$ as a column vector. 

Calculate the proportions $q=p/\sum(p)$ as a row vector. 

Set $A^{\$} = M^{(r+1),p}+dq$.

Set $A = A^{\$}\oslash ({A^{\$}}^Te_m^T)e_n$.

\State{\bf Terminate algorithm}

\end{algorithmic}
\end{algorithm}
We also give a Matlab code example of this algorithm, available in \Cref{app:code}. Note that this algorithm converges extremely quickly, with $r=3$ iterations already giving the solution to 3 decimal places.

It turns out that this algorithm has been discovered previously and is known by many names including the \emph{RAS algorithm}, \emph{raking}, the \emph{biproportional fitting algorithm}, and the \emph{iterative proportional fitting procedure}, see \cite{Stone1942, Lomax2015,lahr2004,Bacharach1965,Bacharach1971,Kruithof1937,Deming1940, Bregman1967} for further details. This algorithm can be applied to more general cases, for example, there is no requirement that the columns of $M$ sum to 1, but instead requiring that no row or column entirely consists of zeros. 

Note that Bezout's theorem tells that if the solution space to \Cref{defn:market-invariantrebal} has dimension 0, then there are at most $2^{(n-1)(m-1)}$ solutions for the reduced solution set (in $\mathbb{C}^{m+n-1}$, as $x_2'=1$). We have shown that that for $m=n=2$ there are precisely $2$ real solutions with $x_2'=1$, and for $(m,n)=(2,3)$ there are at most $3$ real solutions with $x_2'=1$, and there may be only $2$ solutions in some cases. In the case $(m,n)=(2,4)$, there are at most $4$ real solutions, although there may be only 2 in some cases. However, in the $(2,2)$ case, there was only one non-negative solution for $A$ and $A^{\$}$, and we have empirically found the same in the $(2,3),(3,2),(2,4),(4,2)$ cases. It turns out that the following holds, as in \citet[Cor.~1]{Bacharach1965}.

\begin{theorem} \label{thm: bachuniqueness}
For every feasible rebalancing problem $(M,a,p)$ the equations in \Cref{defn:market-invariantrebal} have a unique non-negative solution, and it is therefore given by the market-invariant algorithm.
\end{theorem}
This means the market-invariant algorithm gives this unique solution, and the market-invariant rebalancing process from \Cref{def:marketinvar} is well defined and equivalent to the output of the algorithm. 

Note that \citet[Thm.~3]{Bacharach1965} (see also \citet[\S~2.6]{Senata2006}) determines that the algorithm only converges provided certain conditions hold. These conditions are equivalent to the feasibility conditions from \Cref{lem:feasibilityRebal}. 

In the next section we will show that we can formulate the market-invariant rebalancing process as convex optimisation problem. As the objective function is a strictly concave function (refer to \Citet[\S3.1.4--1.5,\S3.2]{Boyd} for details), then this directly implies \Cref{thm: bachuniqueness}.

\subsubsection{Interpretation as an optimisation program} \label{sec: MIoptimisation}
The market-invariant algorithm can be interpreted as an optimisation program. The following is adapted from \citet[\S~3.2]{lahr2004} as first described in \citet{Uribe1966} using tools from information theory. Here the objective function is known as the \emph{information inaccuracy} of $A^{\$}$, which is also called the Kullbaclk-Leibler divergence (see for example \cite[pg.~552]{Handbook2018}). The program can be written as 
\begin{align} \label{eqn:optimisation}
    &\max \sum_i\sum_jA^{\$}_{ij}\log\left(\frac{A^{\$}_{ij}}{M_{ij}}\right)\\ \nonumber
    \text{subject to } \quad &\sum_iA^{\$}_{ij}=p_j \quad \forall j, \quad \sum_jA^{\$}_{ij}=a_i \quad \forall i, \quad A^{\$}_{ij} \ge 0 \quad \forall i,j.
\end{align}
Note for a given $i,j$ if $M_{ij}$ is equal to 0, we set $A^{\$}_{ij}=0$ and the objective function does not sum over this $i,j$ pair.

Using Lagrange multipliers, we can solve the above optimisation problem. We have the Lagrangian 
\[ \mathcal{L}(A^{\$},\lambda,\nu) = \sum_i\sum_jA^{\$}_{ij}\log\left(\frac{A^{\$}_{ij}}{M_{ij}}\right) + \sum_i \lambda_i(a_i - \sum_jA^{\$}_{ij})+  \sum_j\nu_j(p_j-\sum_i A^{\$}_{ij}),\] with $\lambda\in \R^m,\nu \in \R^n$.
Taking the derivative with respect to $A^{\$}_{ ij}$ and setting equal to zero gives 
\[\log(\frac{A^{\$}_{ij}}{M_{ij}}) +1-\lambda_i-\nu_j =0.\] This implies that 
\[ A^{\$}_{ij} = M_{ij}\exp(-1 + \lambda_i + \nu_j) =\exp(-1 + \lambda_i) M_{ij}\exp(\nu_j),\] for all $i,j$. So we have that $x_i' = \exp(-1 + \lambda_i)$ and $p'_j = \exp( \nu_j)$, using the notation from \Cref{defn:market-invariantrebal}. This shows that the optimisation program is equivalent to the market-invariant rebalancing process. 

As the optimisation program is convex, then convex optimisation algorithms can be used to solve this numerically, as in \citet{Boyd}.

We note that program is equivalent to following
\begin{align*}
    &\max \left ( \prod_{j=1}^n\prod_{i=1}^n\left(\frac{A^{\$}_{ij}}{M_{ij}}\right)^{A^{\$}_{ij}}\right)^{\frac{1}{\sum_{ij}A^{\$}_{ij}}} \\
    \text{subject to } \quad &\sum_iA^{\$}_{ij}=p_j \quad \forall j, \quad \sum_jA^{\$}_{ij}=a_i \quad \forall i, \quad A^{\$}_{ij} \ge 0 \quad \forall i,j.
\end{align*}
Here the objective function is a weighted geometric mean. This looks very similar to Fisher's market model, as in \citet[pg.~106]{Nisan2007}, where unlike our model the objective function is geometric mean of a linear sum of utilities. However, we can equivalently write the optimisation problem as 
\begin{align*}
    &\max \sum_{j=1}^n\log\left(\prod_{i=1}^n\left(\frac{A^{\$}_{ij}}{M_{ij}}\right)^{A^{\$}_{ij}}\right) \\
    \text{subject to } \quad &\sum_iA^{\$}_{ij}=p_j \quad \forall j, \quad \sum_jA^{\$}_{ij}=a_i \quad \forall i, \quad A^{\$}_{ij} \ge 0 \quad \forall i,j.
\end{align*}
Then we could define utility functions $u_{ij}:[0,\infty)\to [0,\infty)$ by 
\[u_{ij}(A^{\$}_{ij}) =\left(\frac{A^{\$}_{ij}}{M_{ij}}\right)^{A^{\$}_{ij}},\] and then define the product utility for each portfolio as 
\[ u_j = \prod_iu_{ij}.\]
Then this is almost in the form of a Eisenberg--Gale-type convex program, as defined in \citet[\S~5]{Jain2010}. This is a class of convex optimisation programs that include Fisher's market model, as in \citet[\S~5.2]{Nisan2007}. Note however that this class of algorithms generally only uses one of the two sets of constraints, and usually relaxed to an inequality, say $\sum_jA^{\$}_{ij}\le a_i$. So while these resource allocation problems are related to this, they are not quite the same.

Finally we can consider the convex dual of \Cref{eqn:optimisation}, which is the following program
\begin{align} \label{eqn:optimisationdual}
    &\min_{\lambda,\nu} \sum_i\lambda_ia_i+\sum_j\nu_jp_j \\ \nonumber
    \text{subject to } \quad & p_j\exp(-\nu_j) = \sum_iM_{ij}\exp(\lambda_i), \quad \forall j,\\
    \quad &a_i\exp(-\lambda_i) = \sum_j M_{ij}\exp(\nu_j), \quad \forall i.
\end{align}
We see that there is an invariance up to scale, were for a solution $\lambda_i,\nu_j$ then for any $r\in \R$ we have that $\lambda_i+r,\nu_j-r$ is also a solution. Making $r$ large enough would ensure that $\lambda_i+r$ is positive for all $i$ and $\nu_j-r$ is negative for all $j$. Then we could interpret this optimisation program as minimising the prices $\lambda_i$ and $\nu_j$ (up to scale) of supplying the assets $a_i$ and portfolios $p_j$ respectively. This is similar to the way the dual variables are considered prices in the Eisenberg--Gale-type convex programs, \cite[\S~5]{Jain2010}. 

Of course, in general we do not expect that the objective function has a minimum of zero, which would only occur in the rare case that the solution is $A^{\$}=M$. If we interpret the $\lambda_i$ and $\nu_j$ as prices, the objective function could then be considered the total cost of the solution for $A^{\$}\ne M$.

\subsubsection{Interpretation of the R and S in the RAS algorithm}
The RAS algorithm is equivalent to our market-invariant algorithm as it seeks to find a matrix $B$ of the form $RAS$ where $R$ and $S$ are square diagonal matrices. Here $B$ corresponds to our $A^{\$}$, while $R$ corresponds to a diagonal matrix with our $x'$ along the diagonal, and $S$ corresponds to a diagonal matrix with our $p'$ along the diagonal. There are some questions around how to interpret $R$ and $S$, with \citet[\S 2.1-2.2]{lahr2004} discussing several alternatives. 

\citet[pg.~23]{Bacharach1971} suggests in particular that while 
\citet{Leontief1941} and \citet{Neisser1941} propose interpreting these matrices as proportional changes in either the asset classes or the portfolios, that this seems implausible in practice.

We think our interpretation in \Cref{defn:market-invariantrebal} is an important addition to this discussion: $S$ consists of portfolio values $p'$ that would allow $Mp'=a'$, i.e. that give a solution $A'$ that is exactly equal to the target $M$. On the other hand, $R$ represents the increase or decrease in asset class values only that, given $a'$ and $p'$, would result in the current $p$ and $a$. 

Note that the invariance under scale property of these solutions would allow us to write $p'$ such that $p'_1=p_1$. Then $p'-p$ could be interpreted as how much the portfolio values would need to change (increasing or decreasing) relative to $p_1$ as a linear difference to get exactly to the target $M$ for each portfolio, and $x'$ represents the same for the asset classes but the proportional increase or decrease. We can also rescale so that $a'_1=a_1$, then $a'-a$ would be how much the asset classes would need to change to get exactly to the target $M$ relative to $a'$ to be at the target. These representations can be helpful for external rebalancing, as this determines how much would be required to trade with the market to return an asset allocation closer to the target $M$.

We should also point out here the symmetry between the asset classes and portfolios: we have used the convention from industry that the matrix $A$ (or $M$) has the column sums equal to $1$, so it represents the proportion of the portfolio value that is made up of each asset class. One can rescale each row of $A$ (or $M$) to $A'$ (and $M'$) so that instead the rows sum to 1, which is then the proportion of each asset class that is allocated to each portfolio. If one then applied all of the same theory following this symmetry, determining the corresponding $a'$, $p'$ and $x'$, we would have $M'^Ta'=p'$, and this $x'$ would represent the proportional increase in each portfolio to achieve $A'=M'$. This is less relevant for our application but may be relevant for interpreting $R$ and $S$ in other areas.

In summary, contrary to \citet{Bacharach1971}, we suggest interpreting these matrices as proportional changes in the asset classes or the portfolio values is appropriate for this application.

\section{Comparisons of different processes} \label{subsec:comparisons}
In this section we compare the performance of the different rebalancing processes as they are used in practice in portfolio management. We start with a theorectical result.

\subsection{Theoretical results}
Recall that we are motivated by the following scenario: a portfolio manager has a series of portfolios with values $p_1,\ldots,p_n$ that use asset classes with values $a_1,\ldots,a_m$. These asset classes need to be fully allocated to the portfolio and each portfolio $j$ has a target proportion $M_{ij}\ge 0$ for asset class $i$. 

At any given time, the manager allocates the asset classes to the different portfolios according to their chosen rebalancing process. Over time, market movements and cash flows change the allocation of each of the portfolios. The manager then reapplies the rebalancing process to rebalance the portfolios closer to their targets --- usually on a weekly or monthly basis. 

Portfolio managers have not previously studied how the different process may advantage or disadvantage their portfolios with respect to market movements. In particular, we note that the banker process is often used as it is said to balance out any advantages or disadvantages to any particular portfolio overtime. The following theorem contradicts this, showing that the banker portfolio is always disadvantaged during `volatile' periods. Specifically, we take any series of returns that result in the same asset allocation as started with and, ignoring all cash flows, show that the banker process is always disadvantaged by this process. This means that market fall and recovery events disadvantage the banker portfolio, as does any period of increased market volatility. This theorem motivates using the market-invariant process instead of the banker or linear process.

\begin{theorem} \label{thm:advantaging}
Let $(M,a,p)$ be a rebalancing problem. Take any series of returns $r_i^t\in (0,\infty)$  for the asset classes $i=1,\ldots,m$ over time periods $t=1,2,\ldots, T$ such that the total return for each asset class over the time period is zero, so 
\[\prod_{t=1}^T (r_{i}^t+1)=1 \quad \text{ for all } i=1,\ldots, m.\] Here we say that the returns are \emph{tethered} so that total return is zero.

For a given rebalancing process, proceed by first applying the rebalancing process to allocate the assets to the portfolios. Then apply the returns for the asset classes for the first time period, resulting in performance of the portfolio corresponding to the ratio of asset classes it contains. Then apply the rebalancing process again to redistribute the assets, then calculate the performance again, and continue through all time periods. Then we have that
\begin{enumerate}
    \item The market-invariant rebalancing process applied at each time period will result in all the portfolios having the same return of 0\%, aligning with the asset classes returning 0\%. No portfolio is advantaged over another.
    \item The linear rebalancing process applied at each time period will result in all the portfolios having slightly different total returns, depending on the order of the returns. This means that some portfolios are advantaged while others are disadvantaged.
    \item The banker rebalancing process applied at each time period will result in the banker portfolio returning negative returns and all other portfolios returning positive returns. That is, the banker is disadvantaged compared to all other portfolios, regardless of the returns. 
\end{enumerate}
We show this matches empirical results.
\end{theorem}
\begin{proof}
For point 1, this results by definition of the market-invariant rebalancing process. In fact, at each time period other than the first, no reallocation is necessary. Empirically using the code in \Cref{app:code}, simulations show that the return is zero to 15 decimal places, which are due to rounding errors. This means that no portfolio is advantaged over another portfolio when using this process.

For point 2 we show this empirically through simulations. Note that if all assets moved the same amount at each time period, no reallocation is necessary. In general however, the return will depend on the order of the returns. This means that using the linear rebalancing process, some portfolios are advantaged and some portfolios are disadvantaged due to market movements.

For point 3 we will take an arbitrary portfolio $j$ such that it is not the banker. We want to consider the ratio \[ \frac{p_j^T}{p_j}\] where $p_j^t$ is the value of the portfolio at time $t$, and $p_j^0=p_j$. If this ratio is greater than 1, this portfolio has made money while less than one means this portfolio has lost money. Then consider that at time $t$
\[p_j^t = \sum_{i=1}^m M_{ij} (r_i^1+1)p_j^{t-1}.\] 
Iterating, we see that we have
\begin{align*}
    \frac{p_j^T}{p_j} &= \frac{\sum_{i_1,\ldots,i_T=1}^m M_{i_1j}M_{i_2j}\cdots M_{i_Tj} (1+r_{i_1}^1)(1+r_{i_2}^2)\cdots (1+r_{i_T}^T) p_j^{0}}{p_j^0} \\
    &= \sum_{i_1,\ldots,i_T=1}^m M_{i_1j}M_{i_2j}\cdots M_{i_Tj} (1+r_{i_1}^1)(1+r_{i_2}^2)\cdots (1+r_{i_T}^T).
     \end{align*}
We want to prove that this is greater than or equal to 1 for this portfolio to not have been disadvantaged through applying the banker process. This is trivially true for $T=1$, as then $(1+r_i^1)=1$ and the performance is exactly equal to 1.

The proof can be derived from using $\prod_{t=1}^T (1+r_{i}^t)=1$ to show
\begin{equation}\label{eqn:conjecturer}
\sum_{\sigma \in P(i_1,i_2,\ldots, i_T)}(1+r_{\sigma_1}^1)(1+r_{\sigma_2}^2)\cdots (1+r_{\sigma_T}^T)\ge  |P(i_1,i_2,\ldots, i_T)|
\end{equation} for all choices of $i_1,\ldots, i_T\in \{1,2,\ldots m\}$. 
Here $P(i_1,i_2,\ldots, i_T)$ is the set of all distinct permutations of $i_1,\ldots, i_T$ and $p=|P(i_1,i_2,\ldots, i_T)|$ is the total number of permutations. Note that this is equal to the multinomial coefficient 
\[ {T \choose k_1,\ldots, k_S}\] 
where $S$ is the number of distinct elements in the list $i_1, \ldots, i_T$, and $k_S$ is the number of times that element appears in the list. For example, if we have $i_1=1, i_2 =1, i_3 =2$ then $S=2$ and $k_1=1$ and $k_2=2$.

We can then prove \Cref{eqn:conjecturer} as we have 
\begin{align*}
    \frac{1}{p}\sum_{\sigma \in P(i_1,i_2,\ldots, i_T)}(1+r_{\sigma_1}^1)(1+r_{\sigma_2}^2)\cdots (1+r_{\sigma_T}^T)&\ge \sqrt[\leftroot{-3}\uproot{5}p]{\prod_{\sigma \in P(i_1,i_2,\ldots, i_T)}(1+r_{\sigma_1}^1)(1+r_{\sigma_2}^2)\cdots(1+ r_{\sigma_T}^T)}\\
    & = \sqrt[\leftroot{-1}\uproot{3}p]{\prod_{i\in \{i_1,\ldots, i_T\}}\prod_{t=1}^T r_{i}^t}\\
    &= \sqrt[\leftroot{-3}\uproot{5}p]{\prod_{i\in \{i_1,\ldots, i_T\}}1}\\
    &= \sqrt[\leftroot{-3}\uproot{3}p]{1}\\
    &= 1,
\end{align*}
where the first line uses that arithmetic means are greater than or equal to geometric means, as in \citet[\S~2.5]{Hardy1988}. Note that equality occurs if and only if all terms of the form \[(1+r_{\sigma_1}^1)(1+r_{\sigma_2}^2)\cdots (1+r_{\sigma_T}^T)\] are equal.

If we write that
\[\mathcal{A}(m,T) = \left\{(i_1,\ldots, i_T) : i_1\le i_2 \le \ldots \le i_T, i_t \in \{1,2,\ldots, m\}\right\},\] then we have that
\begin{align*} 
\sum_{i_1,\ldots,i_T=1}^m& M_{i_1j}M_{i_2j}\cdots M_{i_Tj} (1+r_{i_1}^1)(1+r_{i_2}^2)\cdots (1+r_{i_T}^T) \\ &= \sum_{(i_1,\ldots, i_T)\in \mathcal{A}(m,T)}M_{i_1j}M_{i_2j}\cdots M_{i_Tj}\sum_{\sigma \in P(i_1,i_2,\ldots, i_T)}(1+r_{\sigma_1}^1)(1+r_{\sigma_2}^2)\cdots (1+r_{\sigma_T}^T)\\
& \ge \sum_{(i_1,\ldots, i_T)\in \mathcal{A}(m,T)}M_{i_1j}M_{i_2j}\cdots M_{i_Tj}|P(i_1,i_2,\ldots, i_T)| \quad \text{using \cref{eqn:conjecturer}}\\
&= (\sum_{i=1}^m M_{1j})^T = 1,
\end{align*} as required, where the last line uses the Multinomial Theorem (see \citet{Tauber1963} or \citet{Riordan1978} for details).   

Equality only occurs if all terms $(1+r_{\sigma_1}^1)(1+r_{\sigma_2}^2)\cdots (1+r_{\sigma_T}^T)$ are equal for all possible $\sigma \in P(i_1,i_2,\ldots, i_T)$. This only occurs if every $r_{i_1}^t=r_{i_2}^t$ for all $i_1,i_2$ and for all $t=1,\ldots, T$. This means that unless all asset classes move in unison in every time period, every portfolio except banker portfolio will be advantaged by applying the banker process. 

Unfortunately, this then implies that the banker's performance is 
\[ \frac{\sum_j{p_j} - \sum_{j\ne j_b}{p^T_j}}{\sum_j{p_j} - \sum_{j\ne j_b}{p^0_j}}< 1,\] whenever the asset classes do not move in unison, so that the banker process always disadvantages the banker portfolio. In effect, the banker process is giving returns to the other portfolios. 

In \Cref{fig:empiricalHist} we have empirically tested this result for $t=30$. We see that the banker process results in negative returns for the banker, while the market-invariant process results in zero return (up to rounding errors in order of $10^{-15}$) and the linear process returns in both positive and negative returns.
\end{proof}

\begin{figure}
    \centering
    \includegraphics[width=15cm]{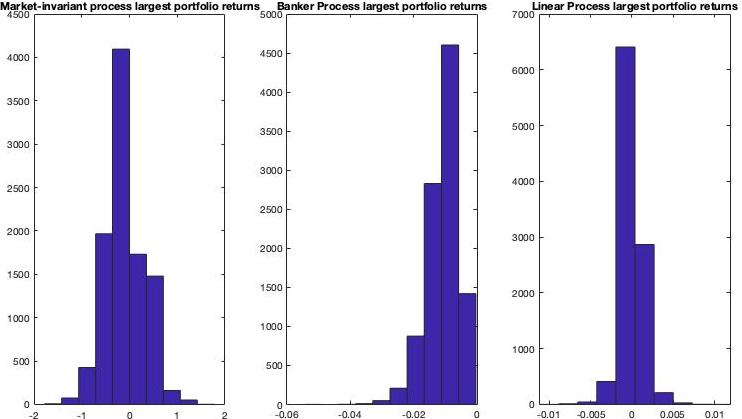}
    \caption{Histograms of the returns of the largest/banker portfolio under 3 different rebalancing processes run over 30 time periods with 10,000 samples. The last two time periods ensure that the total return for each asset class is zero. We observe the market-invariant rebalancing process returning near zero in the order of $10^{-15}$ which are rounding errors, the banker process always returning a negative return, while the linear rebalancing process sometimes returns negative and sometimes returns positive returns. Code for this example is in \Cref{alg:Comparisonsmatlab}.}
    \label{fig:empiricalHist}
\end{figure}

An initial question from this theorem is to what extent is the banker portfolio disadvantaged? Another observation of the previous theorem is that the tethered return series described appears theoretical and unlikely to be observed in practice. In the following section we address these points.

\subsection{Simulation results} \label{rem: advantaging}
We first consider the situation of the previous theorem where the returns are tethered to zero to understand how much the banker portfolio is disadvantaged. To do this, we set up a \emph{shadow} portfolio, which is a portfolio that has the same target asset allocation as the banker. We then want to measure the difference between the returns of the shadow portfolio and the banker portfolio over the time series. As we proved in \Cref{thm:advantaging}, this will always be negative when the returns are tethered unless all the returns are the same. However, we may want to understand what contributes to the extent of the negative returns. 

To do this we plotted the difference between the performance of the banker portfolio and the shadow portfolio against the weighted variance of the return series of the asset classes. The weighting here is by the starting value of the asset classes. Using the banker process, we see that a linear regression model is statistically significant with higher variance resulting in lower performance of the banker compared to the shadow portfolio. The R-squared value is 0.464, suggesting nearly half of the trend in performance is explained by the weighted variance of the asset classes. 
\begin{figure}
    \centering
    \includegraphics[width=15cm]{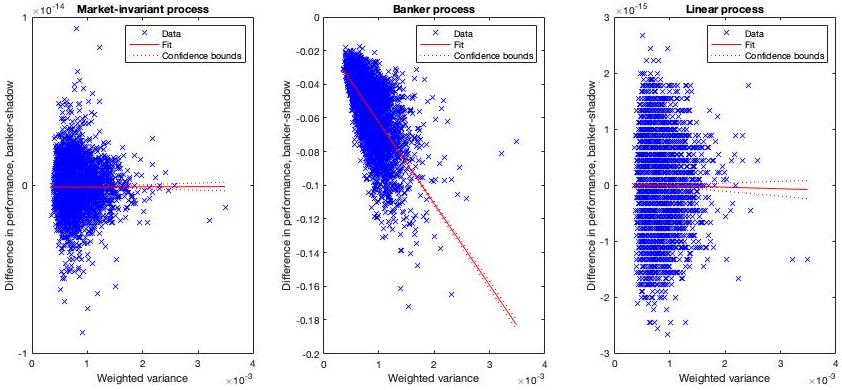}
    \caption{Scatterplots of the difference in performance of the largest/banker portfolio with its shadow portfolio against the weighted variance of the asset classes. The weighting is by starting asset class value. The three different rebalancing processes were run over 30 time periods with 10,000 samples. The last two time periods ensure that the total return for each asset class is zero. Linear regression models are added. The market-invariant and the linear process are not statistically significant, and we note the performance difference is in the order of $10^{-15}$ which are rounding errors. The banker process has a statistically significant linear trend with p-value of less than $10^{-100}$. Initial target asset allocation $M$, return series and initial asset allocations are chosen randomly, and the shadow banker has the same target asset allocation as the banker/largest portfolio.}
    \label{fig:empiricalHist2}
\end{figure}

We are aware that having the returns tethered exactly to zero may seem unlikely to occur in practice. However, we expect this to be illustrative of a market fall and recovery event. More generally, an observed return series could be considered as a decomposition into a volatile tethered series and an increasing or decreasing trend, so there is some merit in considering the volatile tethered series separately. However, the process of rebalancing is not so easily decomposed in this way. 

Running the same analysis on untethered data, so that the end values of the asset classes are random and may be higher or lower than their starting values, we observe similar but more noisy results. There is still a statistically significant linear trend between the difference in performance of the banker and the shadow portfolio against the weighted variance of the asset classes for the banker process. The R-squared value has lowered to 0.0174, suggesting only around $1\%$ of the difference in performance is explained by the variance. Interestingly, the difference in performance is now both positive and negative, although we note that $62\%$ of the time the difference remained negative and the banker portfolio was disadvantaged in these cases. 

\begin{figure}
    \centering
    \includegraphics[width=15cm]{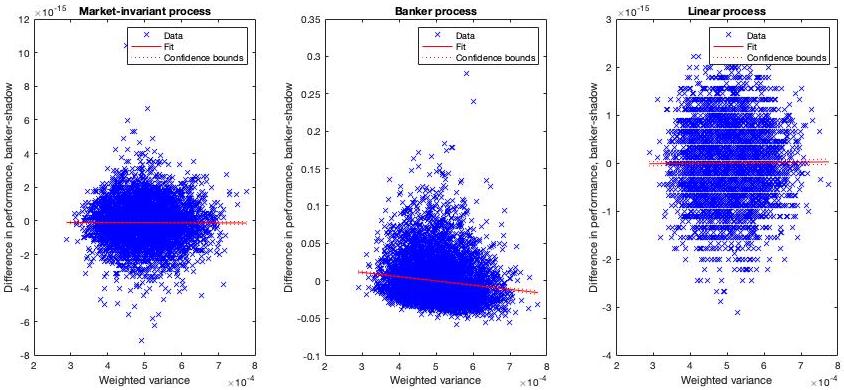}
    \caption{Scatterplots of the difference in performance of the largest/banker portfolio with its shadow portfolio against the weighted variance of the asset classes. The weighting is by starting asset class value. The three different rebalancing processes were run over 30 time periods with 10,000 samples. The returns are untethered so the asset class end values are random and may be higher or lower than the starting values. Linear regression models are added. The market-invariant and the linear process are not statistically significant, and we note the performance difference is in the order of $10^{-15}$ which are rounding errors. The banker process has a statistically significant linear trend with p-value of less than $10^{-100}$. Initial target asset allocation $M$, return series and initial asset allocations are chosen randomly, and the shadow banker has the same target asset allocation as the banker/largest portfolio.}
    \label{fig:empiricalHist3}
\end{figure}

We see that the weighted variance is in the order of $10^{-4}$, yet the difference in performance is between -10 to +30 basis points. This is a sizeable enough difference that funds would want to mitigate to maintain equity between the different portfolios. This is another reason to prefer the market-invariant process where the difference is a rounding error in the order of $10^{-15}$.

Finally, we also analysed whether the weighted performance of the asset classes is also a predictor of the performance differences between the banker and shadow portfolio. We indeed found this is the case, with positive performance contributing to banker outperforming the shadow portfolio. This matches intuition, where increasing markets will benefit the banker by overweighting the banker to the positive performing asset classes. While the banker portfolio is benefited in these markets, the other portfolios are disadvantaged and we see the shadow portfolio under perform. 

The model with the best fit we found to be the following 
\begin{equation} P_b-P_s = 0.032+0.127r+2.159r^2-65.77v\label{eqn:bestfit}\end{equation} where $P_b$ is the performance of the banker, $P_s$ is the performance of the shadow, $v$ is the weighted variance and $r$ is the weight return of the asset classes. The full results are in \Cref{app: regressionresults}. The R-squared is slightly higher at $0.0517$. 

This is the beginning of larger analysis that would be needed to understand exactly how the banker process impacts returns of the different portfolios. However, these results tell us that the banker process does indeed impact the returns of the portfolios. It may advantage or disadvantage the banker portfolio, and increasing the volatility in general disadvantages the banker's performance while increasing performance advantages the banker's performance. 

While neither the market-invariant nor linear process advantage nor disadvantage portfolios with the same target asset allocations, from \Cref{thm:advantaging} we know that in general the linear process is also advantaging or disadvantaging portfolios with different target asset allocations due to market movements. This affirms the recommendation that the market-invariant process be used for internal rebalancing processes to prevent any portfolio being advantaged over another due to market movements.

\section{Conclusion}

While internal rebalancing processes have been used in the financial industry for some time, this paper is the first study of these processes in the literature. We have summarised the important linear and banking processes used in practice, while we have detailed the new market-invariant process. In \S\ref{subsec:comparisons} we have shown in particular how the market-invariant process addresses issues with the banker and linear process in advantaging or disadvantaging portfolios due to market movements. 

One may be concerned that the results in \Cref{thm:advantaging} consisted of a specific theoretical set-up, which may not appear in practice. However, this set-up is very similar to a market fall and recovery. In \Cref{rem: advantaging} we discuss how more general market movements affect the banker portfolio when using a banker process, with increasing volatility further disadvantaging the banker while increasing returns advantaging the banker and disadvantaging other portfolios. By comparison, the market-invariant process does not advantage or disadvantage any portfolios due to market-movements. 

In addition, as the market-invariant process spreads the underweight/overweights to asset classes across all the portfolios, it also spreads the associated risks proportionally to the target asset allocations. This is unlike the banker process which distorts only the banker's asset allocation, and unlike the linear process that tends to over-allocate overweights to portfolios with smaller target allocations. In using the banker or linear process, one then needs to monitor not just the overall overweights/underweights to each asset class, but also how each portfolio is affected by these.

To summarise the benefits of the market-invariant process, they include not advantaging nor disadvantaging portfolios due to market movements, spreading overweight/\allowbreak underweights (including risk) to asset classes proportionally, and always giving the portfolios the same asset allocations if they have the same target asset allocations. 
Given these results, we recommend that the market-invariant process be used for such internal rebalancing processes instead of the linear or banker process. We have left open the question of how cash flows between portfolios, between asset classes, and into or out of the fund will affect rebalancing processes, which we leave for future work. 

\section{Acknowledgements}
The author would like to acknowledge Adjunct Professor Langford White and Adjunct Professor Nicholas Buchdahl from the University of Adelaide for their advice and comments. We also acknowledge Statewide Super's Chief Investment Officer Con Michalakis and Deputy Chief Investment Officer Chris Williams for their understanding of superannuation funds' rebalancing processes and their comments. Additionally, we would like to thank Mr Cameron Baulderstone for the references on optimal power flow.

\newpage

\bibliographystyle{abbrvnat}
\bibliography{ERA_bibliography}
\addcontentsline{toc}{section}{References}

\newpage
\appendix

\section{Electricity networks as rebalancing problems}
\label{app: electricitynetworks}

Electricity networks are related to rebalancing problems. They are related to the supply-demand problems mentioned in \Cref{subsec:supplydemandgeneral}, as the network requires that the the supply of electricity is equal to the demand. 

More specically, in the assumptions used for DC optimal powerflow as in \citet{Bienstock}, the reactance of each line $x_{ij}$ is used to determine the powerflow along the collection of lines $(i,j)$ which start at a node $i$ and end at node $j$. Generation (supply) and demand can be placed at different nodes. 

To make this align with a rebalancing problem $(M,a,p)$, let the starting nodes $i$ be the supply nodes $i=1,\ldots, m$ and the ending notes $j$ be the demand nodes $j=1,\ldots, n$ and say that there is a line between the nodes $i$ and $j$ if $M_{ij}\ne 0$. Then given the \emph{power transfer distribution matrix} $D$ we can determine the (real/active) power $P_{ij}$ along each line $(i,j)$ by writing 
\[P_{ij} = \sum_{k}D_{ij,k}P_k.\]
Here $P_{ij}$ is positive if the power flows from line $i$ to $j$ and negative if from $j$ to $i$. The value $k$ iterates along each supply node $i$ and each demand node $j$ with the injected or withdrawn power at node denoted $P_k$, with injected power being positive $P_i = a_i$ and withdrawn power being negative $P_j = -p_j$. 

The matrix $D$ is determined in terms of the `reactance' $x_{ij}$ (a property of a powerline similar to resistance) and the `line susceptance' and `bus reactance' matrices, which are determined from the geometry of the network as in \citet[\S~3.8]{Chatz2018}. An equivalent formulation is also available in \citet[pg.~13--16]{Bienstock}. 

While the $x_{ij}$ are not required to have columns summing to 1, we can let $x_{ij}=M_{ij}$ and then we would have that $A^{\$}_{ij}=P_{ij}$ is a solution to a rebalancing problem $(M,a,p)$ up to one issue: this does not necessarily return non-negative solutions, as there is no reason electricity cannot flow in either direction. So this is a valid rebalancing process if we relax the condition that $A$ be non-negative.

These calculations are used in optimal powerflow models. They form the constraints used to ensure that supply equals demand when determining optimal electricity generation to minimise costs, as discussed further in \cite{Bienstock,Chatz2018}.

\newpage

\section{Matlab code} \label{app:code}

\begin{algorithm}
\caption{Market-invariant rebalancing matlab code for $m=n=2$} \label{alg:EqtRebalmatlabmn2}
\begin{verbatim}
M=[0.3,0.5;0.7,0.5];
p=[120,180];
a=[100,200];

constant1 = M(1,1)*M(2,1)*(a(1)+a(2))/p(1);
constant2 = M(1,1)*M(2,2)*(a(1)/p(1) - a(2)/p(2)) + ...
                M(2,1)*M(1,2)*(a(2)/p(1)-a(1)/p(2));
constant3 = -M(1,2)*M(2,2)*(a(1)+a(2))/p(2);
scale1=(-constant2+sqrt(constant2^2-4*constant1*constant3))/(2*constant1);
scale2=(-constant2-sqrt(constant2^2-4*constant1*constant3))/(2*constant1);
p1=[scale1,1];
p2=[scale2,1];

%positive solution
pprime=[max(scale1,scale2),1];
%pprime=[min(scale1,scale2),1];
x=zeros(1,2);
x(1) = a(1)/(M(1,:)*pprime');
x(2) = a(2)/(M(2,:)*pprime');
A=x'.*M.*pprime
Ad= A./sum(A)

%negative solution
%pprime=[max(scale1,scale2),1];
pprime=[min(scale1,scale2),1];
xn=zeros(1,2);
xn(1) = a(1)/(M(1,:)*pprime');
xn(2) = a(2)/(M(2,:)*pprime');

An=xn'.*M.*pprime
Adn= An./sum(An)
\end{verbatim}
\end{algorithm}

\begin{algorithm}
\caption{Market-invariant rebalancing Matlab code} \label{alg:EqtRebalmatlab}

\begin{verbatim}
M=[0.3,0.4,0.5,0.1;0.3,0.2,0.3,0.4;0.4,0.4,0.2,0.5];
p=[1030,40,50,60];
a=[55,60,1065]';
ep=ones(1,length(p));
ea=ones(1,length(a))';
r=10; %number of iterations

Mp=M.*p; %sum(Bp) = p = (Bp'*ea)'=ea'*Bp , sum over i
Ma=Mp.*(a./(Mp*ep')); %sum(Ba')=a' = (Ba*ep')', sum over j

for k=1:r
    Mp = Ma.*(p./(Ma'*ea)');
    Ma = Mp.*(a./(Mp*ep'));
    %Ma = Ma.*(p./(Ma'*ea)').*(a./(Ma.*(p./(Ma'*ea)')*ep'));
end

d = a-Mp*ep';%if r is large enough, this difference should be 0
q = p/sum(p); % vector p as a proportion. 

Ad=Mp+q.*d;
A= Ad./(ea'*Ad);
\end{verbatim}
\end{algorithm}

\newpage

\subsection{Comparisons of difference rebalancing processes Matlab code} \label{alg:Comparisonsmatlab}

\begin{verbatim}
clear all
iterations = 10000;
%record the values of the portfolios for the different processes
PerEquit = zeros(4,iterations);
PerBank = zeros(4,iterations);
PerLinear = zeros(4,iterations);
numperiods=30;
startportfolios = [50, 540, 50, 80]; %can randomise
M = [ 0.2  0.05  0.05 0.01;
        0.2  0.05  0.05 0.02;
        0.15 0.25  0.25 0.35;
        0.15 0.30  0.30 0.40;
        0.3  0.35  0.35 0.22];
%startingassets=[100,100,150,170,200]'; %different example
startingassets = M*startportfolios';

for j =1:iterations
    returns =  exp((rand(5,numperiods-2)-0.5)/2);
    returns =[returns,sqrt(1./prod(returns,2)),sqrt(1./prod(returns,2))]-1; %get
    %back to original starting point over last two periods
    periods = size(returns,2);
    returns1 = returns+1;
    assetsovertime = zeros(5,periods+1);
    assetsovertime(1:5,1) = startingassets;
    for i =1:periods
        assetsovertime(1:5,i+1) = returns1(1:5,i).*assetsovertime(1:5,i);
    end
    [Ad1, A1] = marketinvariantrebalancing(startingassets,startportfolios,M);
    
    for i=1:periods
        newAd = Ad1.*returns1(1:5,i);
        currentportfoliovalue = sum(newAd);
        currentassetvalues = sum(newAd')';
        [Ad1 A] = marketinvariantrebalancing(currentassetvalues,currentportfoliovalue,M);
        A-M;
    end
    
    finalEquitiableAd = Ad1;
    finalportfoliovalues = sum(Ad1);
    PerformanceOfPortfoliosmarketinvariant = finalportfoliovalues./startportfolios -1;
    [Ad1, A1] = bankerrebalancing(startingassets,startportfolios,M,2);
    
    for i=1:periods
        newAd = Ad1.*returns1(1:5,i);
        currentportfoliovalue = sum(newAd);
        currentassetvalues = sum(newAd')';
        [Ad1 A] = bankerrebalancing(currentassetvalues,currentportfoliovalue,M,2);
        A-M;
    end
    
    finalBankerAd = Ad1;
    finalBankerportfoliovalues = sum(Ad1);
    PerformanceOfPortfoliosBanker = finalBankerportfoliovalues./startportfolios -1;
    
    [Ad1, A1] = linearrebalancing(startingassets,startportfolios,M);
    
    for i=1:periods
        newAd = Ad1.*returns1(1:5,i);
        currentportfoliovalue = sum(newAd);
        currentassetvalues = sum(newAd')';
        [Ad1 A] = linearrebalancing(currentassetvalues,currentportfoliovalue,M);
        A-M;
    end
    
    finalLinearAd = Ad1;
    finalLinearportfoliovalues = sum(Ad1);
    PerformanceOfPortfoliosLinear = finalLinearportfoliovalues./startportfolios -1;

    PerEquit(:,j)=PerformanceOfPortfoliosmarketinvariant';
    PerBank(:,j) = PerformanceOfPortfoliosBanker';
    PerLinear(:,j) = PerformanceOfPortfoliosLinear';
end
figure(2)
clf
tiledlayout(1,3)
nexttile
hist(PerEquit(2,:))
title('Market-Invariant Rebalancing Largest Portfolio Returns')
nexttile
hist(PerBank(2,:))
title('Banker Rebalancing Largest Portfolio Returns')
nexttile
hist(PerLinear(2,:))
title('Linear Rebalancing Largest Portfolio Returns')

function [Ad, A] = marketinvariantrebalancing(a,p,M) % a as column, p as row
ep=ones(1,length(p));
ea=ones(1,length(a))';

Mp=M.*p; %sum(Bp) = p = (Bp'*ea)'=ea'*Bp , sum over i
Ma=Mp.*(a./(Mp*ep')); %sum(Ba')=a' = (Ba*ep')', sum over j

for k=1:1000
    Mp = Ma.*(p./(Ma'*ea)');
    Ma = Mp.*(a./(Mp*ep'));
end

difference = a-Mp*ep';
%note sum of difference should be 0
proportionp = p/sum(p);
Ad=Mp+proportionp.*difference;
A= Ad./(ea'*Ad);
end

function [Ad, A] = bankerrebalancing(a,p,M,j) 
% a as column, p as row, %j is banker portfolio
A = M;
Ad = A.*p;
Ad(:,j) = Ad(:,j) + (a - sum(Ad')');
A = Ad./sum(Ad);
end

function [Ad, A] = linearrebalancing(a,p,M) % a as column, p as row,
ouweights = (a-M*p')/sum(a);
A = M+ouweights;
Ad = A.*p;
end

\end{verbatim}

\newpage
\section{Regression model output} \label{app: regressionresults}

The following is the output from the linear models of the difference in performance of the banker/largest portfolio and the shadow portfolio, randomising the initial portfolio values over 30 periods with 10,000 trials. The independent variable is the weighted variance. Here the returns are untethered. The first regression output is the from applying the market-invariant rebalancing process at each period, the second from the banker process and the third from the linear process. Details in \Cref{rem: advantaging}.
\begin{verbatim}
mdl1MI = 
Linear regression model:
    y ~ 1 + x1

Estimated Coefficients:
                    Estimate          SE         tStat      pValue 
                   ___________    __________    ________    _______
    (Intercept)     -5.865e-17    1.6412e-16    -0.35736    0.72083
    x1             -4.6593e-15    1.6539e-14    -0.28171    0.77817

Number of observations: 10000, Error degrees of freedom: 9998
Root Mean Squared Error: 1.04e-15
R-squared: 7.94e-06,  Adjusted R-Squared: -9.21e-05
F-statistic vs. constant model: 0.0794, p-value = 0.778
mdl1Bank = 
Linear regression model:
    y ~ 1 + x1

Estimated Coefficients:
                   Estimate       SE         tStat       pValue  
                   ________    _________    _______    __________
    (Intercept)    0.056559    0.0042733     13.235    1.1791e-39
    x1              -5.7347      0.43063    -13.317    4.0488e-40

Number of observations: 10000, Error degrees of freedom: 9998
Root Mean Squared Error: 0.0269
R-squared: 0.0174,  Adjusted R-Squared: 0.0173
F-statistic vs. constant model: 177, p-value = 4.05e-40
mdl1Lin = 
Linear regression model:
    y ~ 1 + x1

Estimated Coefficients:
                    Estimate          SE         tStat      pValue 
                   ___________    __________    ________    _______
    (Intercept)    -7.4269e-17    1.0647e-16    -0.69753    0.48549
    x1              8.0295e-15     1.073e-14     0.74834    0.45427

Number of observations: 10000, Error degrees of freedom: 9998
Root Mean Squared Error: 6.71e-16
R-squared: 5.6e-05,  Adjusted R-Squared: -4.4e-05
F-statistic vs. constant model: 0.56, p-value = 0.454
\end{verbatim}
Here is the output of the larger regression model for the banker portfolio, from \Cref{eqn:bestfit}.
\begin{verbatim}
    mdl1Bank = fitlm([mean(returnassets'.*startingassets)/sum(startingassets);mean(varassets'.*startingassets)/sum(startingassets)]',(PerBank(2,:)-PerBank(3,:))','y ~ x1 + x1^2 + x2')
mdl1Bank = 
Linear regression model:
    y ~ 1 + x1 + x2 + x1^2

Estimated Coefficients:
                   Estimate       SE         tStat       pValue  
                   ________    _________    _______    __________
    (Intercept)    0.031683    0.0022293     14.212    2.1274e-45
    x1              0.12739     0.010811     11.783    7.7317e-32
    x2              -65.776       4.4045    -14.934    6.8492e-50
    x1^2             2.1591      0.26062     8.2845    1.3378e-16

Number of observations: 10000, Error degrees of freedom: 9996
Root Mean Squared Error: 0.0271
R-squared: 0.0517,  Adjusted R-Squared: 0.0514
F-statistic vs. constant model: 182, p-value = 9.38e-115
\end{verbatim}
The following is the output from the linear models of the difference in performance of the banker/largest portfolio and the shadow portfolio, randomising the initial portfolio values over 30 periods with 10,000 trials. The independent variable is the weighted variance. Here the returns are tethered. The first regression output is the from applying the market-invariant rebalancing process at each period, the second from the banker process and the third from the linear process. Details in \Cref{rem: advantaging}.
\begin{verbatim}
mdl1MI = 
Linear regression model:
    y ~ 1 + x1

Estimated Coefficients:
                    Estimate          SE         tStat     pValue 
                   ___________    __________    _______    _______
    (Intercept)    -1.3157e-16    8.6833e-17    -1.5152    0.12976
    x1              1.8365e-15    7.2734e-15    0.25249    0.80066

Number of observations: 10000, Error degrees of freedom: 9998
Root Mean Squared Error: 1.07e-15
R-squared: 6.38e-06,  Adjusted R-Squared: -9.36e-05
F-statistic vs. constant model: 0.0638, p-value = 0.801
mdl1Bank = 
Linear regression model:
    y ~ 1 + x1

Estimated Coefficients:
                   Estimate        SE        tStat     pValue
                   ________    __________    ______    ______
    (Intercept)    0.038083    0.00093605    40.685      0   
    x1              -7.4775      0.078406    -95.37      0   

Number of observations: 10000, Error degrees of freedom: 9998
Root Mean Squared Error: 0.0116
R-squared: 0.476,  Adjusted R-Squared: 0.476
F-statistic vs. constant model: 9.1e+03, p-value = 0
mdl1Lin = 
Linear regression model:
    y ~ 1 + x1

Estimated Coefficients:
                    Estimate          SE         tStat     pValue 
                   ___________    __________    _______    _______
    (Intercept)      3.623e-17    5.3905e-17     0.6721    0.50153
    x1             -3.0591e-15    4.5152e-15    -0.6775    0.49811

Number of observations: 10000, Error degrees of freedom: 9998
Root Mean Squared Error: 6.66e-16
R-squared: 4.59e-05,  Adjusted R-Squared: -5.41e-05
F-statistic vs. constant model: 0.459, p-value = 0.498
\end{verbatim}

\end{document}